\PassOptionsToPackage{dvipsnames}{xcolor}
\documentclass[final]{ieeetj}
\makeatletter

\def\ps@onlypagenumber{%

  \def\@oddhead{}%

  \def\@evenhead{}%

  \def\@oddfoot{\hbox to \textwidth{\hfill\rffont\thepage}}%

  \def\@evenfoot{\hbox to \textwidth{\rffont\thepage\hfill}}%

}

\makeatother
\let\labelindent\relax
\usepackage{enumitem}
\usepackage{amsmath,amsfonts}
\usepackage{algorithmic}
\usepackage{algorithm}
\usepackage{amssymb}
\usepackage{array}
\usepackage[caption=false,font=normalsize,labelfont=sf,textfont=sf]{subfig}
\usepackage{textcomp}
\usepackage{stfloats}
\usepackage{url}
\usepackage{bm}

\usepackage{verbatim}
\usepackage{graphicx}
\usepackage{cite}
\usepackage{todonotes}
\usepackage{pdfcomment}
\usepackage{tikz}
\usepackage{booktabs}
\usepackage{pgfplots}
\pgfplotsset{compat=newest}
\usepackage{pgfplotstable}
\usepackage{siunitx}
\usetikzlibrary{arrows.meta,positioning,calc}

\newlength{\steplabelwidth}
\usepackage{amsthm}
\usepackage{thmtools}
\newtheorem{theorem}{Theorem}

\newtheorem{corollary}{Corollary}
\newtheorem*{problem*}{Problem formulation}
\newtheorem{proposition}{Proposition}
\newtheorem{remark}{Remark}
\theoremstyle{definition}
\newtheorem{example}{Example}
\newlength\figurewidth
\newlength\figureheight
\newcommand{\grayarrayrule}{%
  \multicolumn{20}{@{}c@{}}{%
    \color{gray!60}\leaders\hrule height 0.3pt\hfill\kern0pt%
  }%
}
\def\OJlogo{}
\def\seclogo{}
\AtBeginDocument{\definecolor{tmlcncolor}{cmyk}{0.93,0.59,0.15,0.02}\definecolor{NavyBlue}{RGB}{0,86,125}}
\makeatletter

\let\@revised\@empty

\makeatother

\makeatletter

\usepackage{etoolbox}
\makeatletter
\patchcmd{\@IEEEproof}{\par\noindent}{\ignorespaces}{}{\PatchFailed}
\makeatother

\def\ps@plain{%

  \def\@oddhead{\vbox{\hsize\textwidth\vspace*{-13pt}\vbox to 29pt{\hsize126pt\vspace*{11.5pt}\OJlogo}\hfill\par\vspace*{-2.75pt}\hbox to \textwidth{\vrule width\textwidth height.3pt depth0pt}}}%

  \let\@evenhead\@oddhead

  \def\@oddfoot{\hbox to \textwidth{\hfill\rffont\thepage}}%

  \def\@evenfoot{\hbox to \textwidth{\rffont\thepage\hfill}}%

}

\makeatother
\begin{document}

\title{Ambiguity Analysis and Design of Sparse Arrays via Generalized Vandermonde Rank Conditions}

\receiveddate{}
\reviseddate{}
\accepteddate{}
\publisheddate{}
\doiinfo{}
\markboth{}{Marius Pesavento}
\author{Marius Pesavento,~\IEEEmembership{Senior Member, IEEE}.
  \affil{Technische Universität Darmstadt, 64283 Darmstadt, Germany}
  \corresp{Corresponding author: Marius Pesavento (email: pesavento@ieee.org).}
  \authornote{This work was supported by the German Federal Ministry of Research, Technology and Space through Project ``Open6GHub+'' under Grant 16KIS2407.
\copyright~2026 Marius Pesavento.
  This work is licensed under a Creative Commons Attribution 4.0
  International License.}
}

\begin{abstract}
  Sparse linear arrays obtained by thinning a uniform linear array (ULA)
  achieve large effective apertures with a reduced number of physical
  sensors and have become a key enabling technology across radar, sonar,
  communications, and integrated sensing and communications. The price
  of thinning, however, is the emergence of ambiguities in the array
  manifold: distinct sets of directions of arrival that produce identical
  sensor measurements, precluding unique identification of multiple
  sources. Conventional sparse-array design criteria, based on
  beampattern shaping or estimation-performance optimization, do not
  fully capture how multiple steering vectors interact jointly to produce
  such ambiguities. This paper develops an algebraic framework for the complete multi-source identifiability analysis of thinned ULAs, characterizing all ambiguity sets for arbitrary source number rather than only at the maximum identifiable number. By relating the rank deficiency of the generalized Vandermonde matrix
  associated with the sparse steering matrix to that of a thinned
  Toeplitz matrix, and further to a rank condition on an augmented
  full-ULA steering matrix with prescribed generators, we obtain a
  systematic characterization of the ambiguity sets in large sparse
  arrays together with constructive design guidelines for
  ambiguity-free geometries. Algebraic and numerical examples demonstrate that the proposed framework
  resolves the ambiguity structure of thinned arrays of a few dozen grid
  positions, for which previous methods yield only individual ambiguity sets.
\end{abstract}
\maketitle
\IEEEpubid{%
  \footnotesize
  \copyright~2026 Marius Pesavento.
  This work is licensed under a
  \href{https://creativecommons.org/licenses/by/4.0/}
  {Creative Commons Attribution 4.0 International License}.%
}
\pagestyle{empty}
\thispagestyle{onlypagenumber}
\pagestyle{onlypagenumber}
\section{Introduction}
Sparse arrays have become a key enabling technology in modern sensing and
communication systems, including automotive and airborne radar
\cite{9429942Sun}, UAV-based sensing systems \cite{11318848UAVsensing},
mmWave/THz massive MIMO communications \cite{10476611pipi}, integrated sensing
and communication (ISAC) systems~\cite{10942827RajamaekiPal,10930389Li}, sonar \cite{10368436wage},
medical ultrasound, radio-astronomy~\cite{8988146Syeda},
and seismic exploration~\cite{1449208capon}, where large effective apertures
are required under tight hardware, cost, and power
constraints~\cite{amin2024sparse,pesaventoThreeMoreDecades2023,10636507liu}.
Sparse linear arrays obtained by thinning uniform linear arrays (ULAs) are
particularly attractive, as they realize large effective apertures with a
substantially reduced number of physical
sensors~\cite{5456168PalPP,5164990Oliveri,5993514Rocca}. For a fixed sensor
count, they provide significantly larger apertures than ULAs; for a fixed
aperture, they reduce hardware cost and complexity. Furthermore, by exploiting
the virtual difference coarray, sparse arrays offer increased degrees of
freedom, permitting the identification of more sources than physical sensors,
while alleviating mutual coupling effects associated with densely packed
arrays~\cite{7458883LiuPP,7458907LiuPP,8061016WangAmin,10835190gini}.

In its simplest form, a sparse linear array is constructed from an underlying
ULA by omitting a subset of sensors while retaining the uniform sampling
grid, resulting in a non-uniform sensor placement on an originally uniform
grid. The larger effective aperture sharpens the main lobe relative to a full
array with the same number of sensors and enhances angular resolution in
direction-of-arrival (DoA) estimation and beamforming. At the same time, the increased inter-sensor spacings raise the sidelobe levels~\cite{10477080Friedlander} relative to a full array with the same aperture and, in the narrowband case, violate the half-wavelength spacing required by the spatial sampling theorem, giving rise to higher-order ambiguities in the array manifold~\cite{abramovichIdentifiabilityManifoldAmbiguity1999}.
\subsection{Related work}
Conventional design and analysis tools for sparse arrays assess performance
through single-source or few-source models. Beampattern-based criteria capture
only the magnitude of the single-source response and neglect the phase
relationships among multiple steering
vectors~\cite{10249750Zhao,5164990Oliveri,8061016WangAmin,wu_sparse_2024}.
Cram\'er--Rao bound driven sensor
selection~\cite{10890202Ids,10094942Ids,9398546Bilik} optimizes the estimation
variance, but presumes a fixed and small number of sources, so as to avoid the
dependence on the source scenario. Coarray-based thinning criteria, such as the
effective aperture, the degrees of freedom, hole-freeness or the length of the
consecutive lag set, and the mutual coupling profile, are surrogates that yield
only lower bounds on the number of uniquely identifiable
sources~\cite{10835190gini}.

None of these criteria characterizes how several steering vectors interact
jointly, and hence under which conditions the array manifold matrix drops rank.
This limitation becomes critical for thinned arrays, where the non-uniform
sampling of the aperture introduces structured dependencies among the steering
vectors that are invisible to single-source or magnitude-only metrics, and that
manifest themselves only when several sources are present simultaneously.

This paper is therefore devoted to the multi-source analysis of sparse linear
arrays, in which higher-order ambiguities arise, i.e., distinct sets of DoA
parameters that yield identical measurements and thus preclude unique DoA
estimation. Wax and Ziskind~\cite{waxUniqueLocalizationMultiple1989} showed that,
for uncorrelated sources, the maximum number of sources that can be uniquely
identified is governed by the Kruskal rank of the array steering matrix, i.e.,
the largest integer~$k$ such that any~$k$ steering vectors in the field of view are
linearly independent. 

Manikas and
Proukakis~\cite{manikasModelingEstimationAmbiguities1998} derived sufficient
conditions for higher-order ambiguities, implying that every thinned ULA suffers
from ambiguities~\cite{abramovichIdentifiabilityManifoldAmbiguity1999,
abramovichStabilityManifoldAmbiguity2000}, and characterized so-called uniform
ambiguity sets, whose generators are roots of unity and therefore correspond to
a uniform partitioning of the field of view in the electrical angle. In the special case where the number of sources equals the number of sensors,
such generator sets can make two rows of the generalized Vandermonde matrix
identical, so that it is rank deficient. A second class of non-uniform ambiguities arises
only in symmetric thinned arrays~\cite{Manikas2004DifferentialGeometryArrayProcessing,
dowlutExtendedAmbiguityCriterion2002}, where the symmetry of the layout is
exploited to construct a centro-Hermitian polynomial whose roots lie on the unit
circle and thus form an admissible generator set.

Conditions for the manifold matrix to have full column rank were derived
in~\cite{9723150ChenPP} in terms of greatest common divisors of the sensor
positions and their differences. They are necessary in general and reduce, for
three and four sensors, to conditions that are also sufficient. Two
restrictions apply: the analysis covers only the maximal case of one source
fewer than sensors, and the ambiguous configurations it yields are equispaced
in the electrical angle.

More recently, in~\cite{matterAmbiguitiesDoAEstimation2022} we proposed a
computational framework that characterizes arbitrary ambiguity sets, including
continuous families, by searching directly for sets of directions at which the
array steering matrix becomes rank deficient, formulated as a mixed-integer
optimization problem. Its computational complexity, however, does not scale well with
the array size, so that in practice the analysis is limited to comparatively
small arrays~\cite{10477080Friedlander}.
\subsection{Contributions}
This paper introduces a framework for identifying ambiguities in thinned ULAs
that reparameterizes the generator set by the coefficients of its annihilating
polynomial, thereby replacing the rank condition on the sparse array manifold by
a rank condition on a complementary Toeplitz submatrix. The framework is computationally efficient and enables, for the first time, the complete characterization of all ambiguity sets of a thinned array for an arbitrary source number, rather than the construction of individual
ambiguities. This makes it practically applicable to the multi-source design of sparse uniform linear arrays.

Beyond array processing, the underlying mathematical problem is of
independent interest. It amounts to determining the spark and the
corresponding generator sets of generalized Vandermonde matrices obtained by
deleting prescribed rows from a full Vandermonde matrix \cite{Heineman1929GeneralizedVandermonde,buck_generalized_1992,SchlickeweiViola2000GeneralizedVandermonde}. Such problems arise
not only in super-resolution spectral estimation \cite{6576276Recht,5713264BajwaEldar,10177164MulletiEldar}, but
also in sparse recovery \cite{books/daglib/0036092Rauhut}, finite-rate-of-innovation sampling \cite{7744614Ye},
algebraic coding theory \cite{1336783Lacan}, sparse polynomial interpolation \cite{joszSparsePolynomialInterpolation2019}, moment problems
\cite{mourrainPolynomialExponentialDecomposition2018}, the observability and controllability of linear systems \cite{fusterSabaterJointCriterion1991}, and multiuser communication \cite{6510017Cardoso,771047Scalione}. In all of these
settings, rank drops of structured Vandermonde-type matrices determine
uniqueness, identifiability, or the existence of ambiguous parameter
configurations. Thus, the results developed in this paper provide both a
practical tool for sparse-array design and a contribution to the broader
theory of structured low-rank  and generalized Vandermonde matrices.

The major contributions of this paper are as follows:
\begin{itemize}
  \item By introducing a convenient null space characterization for the thinned
    array manifold, we relate the rank drop of the generalized Vandermonde matrix
    associated with the sparse array steering matrix to the rank drop of a thinned
    Toeplitz matrix; see Theorem~\ref{thm:rank_condition}\,. This allows the
    characterization of ambiguity sets for sparse arrays of arbitrary aperture size,
    number of sparse array elements, and assumed number of sources.
  \item The resulting criterion is a rank condition on a matrix of size
    $(M-L) \times (M-M_{\mathcal{I}})$ whose entries depend linearly on the
    polynomial coefficients, so that it is evaluated at a cost that grows only
    polynomially in the aperture. It also yields new insights into the formation of
    ambiguities. In particular, it relates the rank condition of the array steering
    matrix of the thinned array to the rank condition of an augmented steering
    matrix of a full ULA; see Corollary~\ref{crl:corollary5}\,.
  \item Building on these criteria, guidelines are developed for the design of
    ambiguity-free thinned arrays in the form of simple necessary conditions on the
    sensors to be retained or removed in the sparse layout. These conditions are
    purely combinatorial and can significantly reduce the search space in sparse
    array designs.
  \item We provide extensive examples for the analysis and design of thinned
    arrays and for the characterization of their ambiguity sets, using both symbolic
    computations and numerical optimization methods. They provide a complete characterization of the ambiguity structure for apertures of a few dozen grid positions, compared with at most individual ambiguity sets obtained by previous methods.
\end{itemize}
\subsection{Paper structure}
The remainder of the paper is organized as follows. Section~\ref{sec:ProblemFormulation} formulates the ambiguity problem for sparse linear arrays. Section~\ref{sec:MainResults} presents an algebraic characterization of the ambiguities and structural insights into ambiguity sets of thinned arrays. 
Section~\ref{sec:ArrayDesign} develops sensor-selection guidelines, obtained by
evaluating the rank condition for particular coefficient vectors, which
substantially reduce the search space of the underlying combinatorial problem. Section~\ref{sec:Examples} illustrates the theory, fully characterizes ambiguity sets for representative arrays and given source numbers, and applies the design principles to construct sparse arrays with favorable ambiguity behavior. Section~\ref{sec:ConclusionAndFutureWork} concludes the paper and outlines open problems amenable to the developed analysis tools.
\subsection{Notation}
Scalars, vectors, and matrices are denoted by $x$, $\boldsymbol{x}$, and
$\boldsymbol{X}$, with $x_i$ and $({\boldsymbol X})_{ij}$ denoting their
respective elements, and $\boldsymbol{x}j=[\boldsymbol{X}]{:,j}$ the
$j$-th column of $\boldsymbol{X}$. The operators $\mathcal{R}(\cdot)$,
$\mathcal{N}(\cdot)$, and $\dim(\cdot)$ denote the range, null space, and
subspace dimension; $\operatorname{rank}(\cdot)$ and $\det(\cdot)$ the
matrix rank and determinant; and $\operatorname{Re}(\cdot)$ and
$\operatorname{Im}(\cdot)$ the real and imaginary parts. The $M\times L$
Vandermonde matrix with generators $z_1,\ldots,z_L$ is denoted by
$\boldsymbol{A}_M(z_1,\ldots,z_L)$. Finally,
$\mathbb{B}^{p\times q}$, 
$\mathbb{R}^{p\times q}$, and $\mathbb{C}^{p\times q}$ denote the sets of
binary, integer, real, and complex $p\times q$ matrices, respectively.

\section{Ambiguities in Sparse Arrays}\label{sec:ProblemFormulation}
The array response matrix of a thinned uniform linear array (ULA) is a
generalized Vandermonde matrix obtained from a full Vandermonde matrix with
distinct generators on the unit circle by removing the rows corresponding
to the missing sensor positions.\,\footnote{The theory of generalized
  Vandermonde matrices developed in this paper extends naturally to arbitrary
generators, not necessarily located on the unit circle.}

Let ${\boldsymbol A}_M(z_1,\ldots,z_L) \in \mathbb{C}^{M \times L}$ denote
the steering matrix of a full ULA with $M$ elements and $L$ sources. It is
a Vandermonde matrix with distinct generators $z_i = e^{{\mathsf j}\theta_i}$,
$i = 1,\ldots,L$, where $\theta_i \in [-\pi,\pi)$, so that
$({\boldsymbol A}_M)_{m,i} = z_i^{m-1}$. In the array processing context,
$\theta$ denotes the spatial frequency (or electrical angle), given by
$\theta = \tfrac{2\pi d}{\lambda} \cos\vartheta$, where $\vartheta$ is the azimuth
angle, $\lambda$ the wavelength, and $d$ the inter-element spacing of the
underlying ULA. A common choice is $d=\lambda/2$, which ensures a
bijective mapping from the azimuth angle $\vartheta$ to the spatial frequency
$\theta \in [-\pi,\pi)$.

Furthermore, let  ${\boldsymbol A}_{\mathcal{I}}(z_1,\ldots,z_L) \in
\mathbb{C}^{M_{\mathcal{I}} \times L}$ denote the steering matrix of the
thinned ULA with the same total aperture $M$ but only
$M_{\mathcal{I}} < M$ sensors retained. The index set
\begin{equation}\label{def:DefIndexSet}
  \!\! \mathcal{I}\! \triangleq \! \{m_1, m_2, \dots, m_{M_{\mathcal{I}}}\},
  \  1\! =\! m_1\! <\! m_2\! < \!\cdots\! <\! m_{M_{\mathcal{I}}}\! = \! M
\end{equation}
contains the indices of the retained sensors in the full ULA. The
corresponding selection matrix ${\boldsymbol T}_{\mathcal{I}}$ is defined as
\begin{equation}\label{def:DefSelectionMatrix}
  {\boldsymbol T}_{\mathcal{I}} \in \mathbb{B}^{M_{\mathcal{I}} \times M},
  \qquad
  \big({\boldsymbol T}_{\mathcal{I}}\big)_{k,m} \triangleq
  \begin{cases}
    1, & m = m_k,\\[4pt]
    0, & \text{otherwise},
  \end{cases}
\end{equation}
so that
\begin{equation}\label{def:DefAtilde}
  {\boldsymbol A}_{\mathcal{I}}(z_1,\ldots,z_L)
  = {\boldsymbol T}_{\mathcal{I}}\,{\boldsymbol A}_M(z_1,\ldots,z_L).
\end{equation}

Rank deficiency of the steering matrix is the algebraic expression of
non-identifiability: if
$\operatorname{rank}({\boldsymbol A}_{\mathcal{I}}(z_1,\ldots,z_L)) < L$, the
columns of ${\boldsymbol A}_{\mathcal{I}}$ are linearly dependent, so that
distinct source configurations and corresponding amplitude vectors exist that
generate the same noise-free array output, regardless of the number of
snapshots.

The maximum number of sources that can be
uniquely localized from array measurements is accordingly governed by the
Kruskal rank of the steering
matrix~\cite{waxUniqueLocalizationMultiple1989}, i.e., by the largest $L$ such
that \emph{every} admissible generator set of size $L$ yields a full-column-rank
manifold.

The Vandermonde matrix ${\boldsymbol A}_M(z_1,\ldots,z_L)$ admits a
convenient characterization of its range as the null space of a structured
(banded, lower-triangular block) Toeplitz matrix built from the coefficients of
the annihilating polynomial. This parameterization goes back to the linear
prediction and annihilating-filter literature and underlies the IQML family of
estimators; see~\cite[Sec.~II]{Bresler1164949} and~\cite[Sec.~III]{57542Stoica}.
Specifically, there exists a Toeplitz matrix
${\mathcal{T}}(\boldsymbol{g}) \in \mathbb{C}^{(M-L)\times M}$ with
parameter vector
$\boldsymbol{g} = [g_0, g_1,\ldots,g_L]^{\mathsf T} \in \mathbb{C}^{L+1}$,
$g_0 \neq 0$, whose rows are orthogonal to  the columns of
${\boldsymbol A}_M(z_1,\ldots,z_L)$, i.e.,
\begin{equation}\label{eq:subspace_property}
  {\mathcal{T}}(\boldsymbol{g})\,{\boldsymbol A}_M(z_1,\ldots,z_L)
  = {\boldsymbol 0}.
\end{equation}
Property~\eqref{eq:subspace_property} states the inclusion
${\cal R}({\boldsymbol A}_M) \subseteq {\cal N}({\mathcal{T}}(\boldsymbol{g}))$.
Equality of the two subspaces, i.e.,
\begin{equation}\label{eq:subspace_property_equality}
  {\cal R}({\boldsymbol A}_M(z_1,\ldots,z_L)) = {\cal N}({\mathcal{T}}(\boldsymbol{g})),
\end{equation}
holds whenever the $L$ generators are \emph{distinct}, since both subspaces then
have dimension $L$. Since the generators are assumed distinct throughout,
\eqref{eq:subspace_property} and~\eqref{eq:subspace_property_equality} are used
interchangeably in the sequel.

For even source number $L = 2K$, the Toeplitz matrix is given
in~\eqref{def:ToeplitzEven} at the top of the next page, with complex
coefficients $g_0, g_1, \ldots, g_{K-1}$ and real coefficient $g_K$. For
odd source number $L = 2K+1$, it is given in~\eqref{def:ToeplitzOdd}, with
complex coefficients $g_0, g_1, \ldots, g_K$.
\begin{figure*}[!t]
  \normalsize
  \begin{align}\label{def:ToeplitzEven}
    \overset{
      \raisebox{3.0ex}{%
        $
        \begin{gathered}
          \text{\normalsize \textcolor{gray}{even:}}\\[-1pt]
          {\normalsize \textcolor{gray}{L\!=\!2K}}
        \end{gathered}$%
      }
    }
    {
    {\mathcal{T}}\!(\boldsymbol{g})}& \!\! \triangleq \!\!
    {
      \setlength{\arraycolsep}{2pt}
      \renewcommand{\arraystretch}{0.75}
      \left[
        \begin{array}{cccccccccccccccccccc}
          \multicolumn{1}{c}{\scriptscriptstyle \textcolor{red}{1}}
          & \multicolumn{1}{c}{\scriptscriptstyle \textcolor{red}{2}}
          & \multicolumn{1}{c}{\scriptscriptstyle \textcolor{red}{3}}
          & \multicolumn{1}{c}{\scriptscriptstyle \textcolor{red}{4}}
          & \multicolumn{1}{c}{\scriptscriptstyle \textcolor{red}{\cdots}}
          & \multicolumn{1}{c}{\scriptscriptstyle \textcolor{red}{K\text{\!--\!}2}}
          & \multicolumn{1}{c}{\scriptscriptstyle \textcolor{red}{K\text{\!--\!}1}}
          & \multicolumn{1}{c}{\scriptscriptstyle \textcolor{red}{K}}
          & \multicolumn{1}{c}{\scriptscriptstyle \textcolor{red}{K\text{\!+\!}1}}
          & \multicolumn{1}{c}{\scriptscriptstyle \textcolor{red}{K\text{\!+\!}2}}
          & \multicolumn{1}{c}{\scriptscriptstyle \textcolor{red}{K\text{\!+\!}3}}
          & \multicolumn{1}{c}{\scriptscriptstyle \textcolor{red}{K\text{\!+\!}4}}
          & \multicolumn{1}{c}{\scriptscriptstyle\textcolor{red}{\cdots}}
          & \multicolumn{1}{c}{\scriptscriptstyle \textcolor{red}{L\text{\!--\!}2}}
          & \multicolumn{1}{c}{\scriptscriptstyle \textcolor{red}{L\text{\!--\!}1}}
          & \multicolumn{1}{c}{\scriptscriptstyle \textcolor{red}{L}}
          & \multicolumn{1}{c}{\scriptscriptstyle\textcolor{red}{L\text{\!+\!}1}}
          & \multicolumn{1}{c}{\scriptscriptstyle \textcolor{red}{L\text{\!+\!}2}}
          & \multicolumn{1}{c}{\scriptscriptstyle \textcolor{red}{\cdots}}
          & \multicolumn{1}{c}{\scriptscriptstyle \textcolor{red}{M}} \\[-6pt]
          \grayarrayrule \\[-1pt]
          g_0 & g_1 & g_2 & g_3 & \cdots
          & g_{K-3} & g_{K-2} & g_{K-1} & g_K
          & g_{K-1}^* & g_{K-2}^* & g_{K-3}^*
          & \cdots & g_3^* & g_2^* & g_1^* & g_0^*
          & 0 & \cdots & 0
          \\
          0 & g_0 & g_1 & g_2 & g_3
          & \cdots & g_{K-3} & g_{K-2} & g_{K-1}
          & g_K & g_{K-1}^* & g_{K-2}^*
          & g_{K-3}^* & \cdots & g_3^* & g_2^* & g_1^*
          & g_0^* & \ddots & \vdots
          \\
          \vdots & \ddots & \ddots & \ddots & \ddots
          & \ddots & \ddots & \ddots & \ddots
          & \ddots & \ddots & \ddots
          & \ddots & \ddots & \ddots & \ddots & \ddots
          & \ddots & \ddots & 0
          \\
          0 & \cdots & 0 & g_0 & g_1
          & g_2 & g_3 & \cdots & g_{K-3}
          & g_{K-2} & g_{K-1} & g_K
          & g_{K-1}^* & g_{K-2}^* & g_{K-3}^*
          & \cdots & g_3^* & g_2^* & g_1^* & g_0^*
          \\[-4pt]
          \grayarrayrule \\[-2pt]
          \multicolumn{1}{c}{\scriptscriptstyle \textcolor{red}{1}}
          & \multicolumn{1}{c}{\scriptscriptstyle \textcolor{red}{\cdots}}
          & \multicolumn{1}{c}{\scriptscriptstyle \textcolor{red}{M\text{\!--\!}L\text{\!--\!}1}}
          & \multicolumn{1}{c}{\scriptscriptstyle \textcolor{red}{M\text{\!--\!}L}}
          & \multicolumn{1}{c}{\scriptscriptstyle \textcolor{red}{M\text{\!--\!}L\text{\!+\!}1}}
          & \multicolumn{1}{c}{\scriptscriptstyle \textcolor{red}{M\text{\!--\!}L\text{\!+\!}2}}
          & \multicolumn{1}{c}{\scriptscriptstyle \textcolor{red}{M\text{\!--\!}L\text{\!+\!}3}}
          & \multicolumn{1}{c}{\scriptscriptstyle \textcolor{red}{\cdots}}
          & \multicolumn{1}{c}{\scriptscriptstyle \textcolor{red}{M\text{\!--\!}K\text{\!--\!}2}}
          & \multicolumn{1}{c}{\scriptscriptstyle \textcolor{red}{M\text{\!--\!}K\text{\!--\!}1}}
          & \multicolumn{1}{c}{\scriptscriptstyle \textcolor{red}{M\text{\!--\!}K}}
          & \multicolumn{1}{c}{\scriptscriptstyle \textcolor{red}{M\text{\!--\!}K\text{\!--\!}1}}
          & \multicolumn{1}{c}{\scriptscriptstyle \textcolor{red}{M\text{\!--\!}K\text{\!+\!}2}}
          & \multicolumn{1}{c}{\scriptscriptstyle \textcolor{red}{M\text{\!--\!}K\text{\!+\!}3}}
          & \multicolumn{1}{c}{\scriptscriptstyle \textcolor{red}{M\text{\!--\!}K\text{\!+\!}4}}
          & \multicolumn{1}{c}{\scriptscriptstyle \textcolor{red}{\cdots}}
          & \multicolumn{1}{c}{\scriptscriptstyle \textcolor{red}{M\text{\!--\!}3}}
          & \multicolumn{1}{c}{\scriptscriptstyle \textcolor{red}{M\text{\!--\!}2}}
          & \multicolumn{1}{c}{\scriptscriptstyle \textcolor{red}{M\text{\!--\!}1}}
          & \multicolumn{1}{c}{\scriptscriptstyle \textcolor{gray}{M}}
        \end{array}
      \right]
    }\\
    \label{def:ToeplitzOdd}
    \overset{
      \raisebox{3.0ex}{%
        $
        \begin{gathered}
          \text{\normalsize \textcolor{gray}{odd:}}\\[-1pt]
          {\normalsize \textcolor{gray}{L\!=\!2K\!+\!1}}
        \end{gathered}$%
      }
    }
    {
    {\mathcal{T}}\!(\boldsymbol{g})} & \! \! \triangleq\! \!
    {
      \setlength{\arraycolsep}{2pt}
      \renewcommand{\arraystretch}{0.75}
      \left[
        \begin{array}{cccccccccccccccccccc}
          \multicolumn{1}{c}{\scriptscriptstyle \textcolor{red}{1}}
          & \multicolumn{1}{c}{\scriptscriptstyle \textcolor{red}{2}}
          & \multicolumn{1}{c}{\scriptscriptstyle \textcolor{red}{3}}
          & \multicolumn{1}{c}{\scriptscriptstyle \textcolor{red}{4}}
          & \multicolumn{1}{c}{\scriptscriptstyle \textcolor{red}{\cdots}}
          & \multicolumn{1}{c}{\scriptscriptstyle \textcolor{red}{K\text{\!--\!}2}}
          & \multicolumn{1}{c}{\scriptscriptstyle \textcolor{red}{K\text{\!--\!}1}}
          & \multicolumn{1}{c}{\scriptscriptstyle \textcolor{red}{K}}
          & \multicolumn{1}{c}{\scriptscriptstyle \textcolor{red}{K\text{\!+\!}1}}
          & \multicolumn{1}{c}{\scriptscriptstyle \textcolor{red}{K\text{\!+\!}2}}
          & \multicolumn{1}{c}{\scriptscriptstyle \textcolor{red}{K\text{\!+\!}3}}
          & \multicolumn{1}{c}{\scriptscriptstyle \textcolor{red}{K\text{\!+\!}4}}
          & \multicolumn{1}{c}{\scriptscriptstyle\textcolor{red}{\cdots}}
          & \multicolumn{1}{c}{\scriptscriptstyle \textcolor{red}{L\text{\!--\!}2}}
          & \multicolumn{1}{c}{\scriptscriptstyle \textcolor{red}{L\text{\!--\!}1}}
          & \multicolumn{1}{c}{\scriptscriptstyle \textcolor{red}{L}}
          & \multicolumn{1}{c}{\scriptscriptstyle\textcolor{red}{L\text{\!+\!}1}}
          & \multicolumn{1}{c}{\scriptscriptstyle \textcolor{red}{L\text{\!+\!}2}}
          & \multicolumn{1}{c}{\scriptscriptstyle \textcolor{red}{\cdots}}
          & \multicolumn{1}{c}{\scriptscriptstyle \textcolor{red}{M}} \\[-6pt]
          \grayarrayrule \\[-1pt]
          g_0 & g_1 & g_2 & g_3 & \cdots
          & g_{K-3} & g_{K-2} & g_{K-1} & g_K
          & g_K^* & g_{K-1}^* & g_{K-2}^*
          & \cdots & g_3^* & g_2^* & g_1^* & g_0^*
          & 0 & \cdots & 0
          \\
          0 & g_0 & g_1 & g_2 & g_3
          & \cdots & g_{K-3} & g_{K-2} & g_{K-1}
          & g_K & g_K^* & g_{K-1}^*
          & g_{K-2}^* & \cdots & g_3^* & g_2^* & g_1^*
          & g_0^* & \ddots & \vdots
          \\
          \vdots & \ddots & \ddots & \ddots & \ddots
          & \ddots & \ddots & \ddots & \ddots
          & \ddots & \ddots & \ddots
          & \ddots & \ddots & \ddots & \ddots & \ddots
          & \ddots & \ddots & 0
          \\
          0 & \cdots & 0 & g_0 & g_1
          & g_2 & g_3 & \cdots & g_{K-3}
          & g_{K-2} & g_{K-1} & g_K
          & g_K^* & g_{K-1}^* & g_{K-2}^*
          & \cdots & g_3^* & g_2^* & g_1^* & g_0^*
          \\[-4pt]
          \grayarrayrule \\[-2pt]
          \multicolumn{1}{c}{\scriptscriptstyle \textcolor{red}{1}}
          & \multicolumn{1}{c}{\scriptscriptstyle \textcolor{red}{\cdots}}
          & \multicolumn{1}{c}{\scriptscriptstyle \textcolor{red}{M\text{\!--\!}L\text{\!--\!}1}}
          & \multicolumn{1}{c}{\scriptscriptstyle \textcolor{red}{M\text{\!--\!}L}}
          & \multicolumn{1}{c}{\scriptscriptstyle \textcolor{red}{M\text{\!--\!}L\text{\!+\!}1}}
          & \multicolumn{1}{c}{\scriptscriptstyle \textcolor{red}{M\text{\!--\!}L\text{\!+\!}2}}
          & \multicolumn{1}{c}{\scriptscriptstyle \textcolor{red}{M\text{\!--\!}L\text{\!+\!}3}}
          & \multicolumn{1}{c}{\scriptscriptstyle \textcolor{red}{\cdots}}
          & \multicolumn{1}{c}{\scriptscriptstyle \textcolor{red}{M\text{\!--\!}K\text{\!--\!}3}}
          & \multicolumn{1}{c}{\scriptscriptstyle \textcolor{red}{M\text{\!--\!}K\text{\!--\!}2}}
          & \multicolumn{1}{c}{\scriptscriptstyle \textcolor{red}{M\text{\!--\!}K\text{\!--\!}1}}
          & \multicolumn{1}{c}{\scriptscriptstyle \textcolor{red}{M\text{\!--\!}K}}
          & \multicolumn{1}{c}{\scriptscriptstyle \textcolor{red}{M\text{\!--\!}K\text{\!+\!}1}}
          & \multicolumn{1}{c}{\scriptscriptstyle \textcolor{red}{M\text{\!--\!}K\text{\!+\!}2}}
          & \multicolumn{1}{c}{\scriptscriptstyle \textcolor{red}{M\text{\!--\!}K\text{\!+\!}3}}
          & \multicolumn{1}{c}{\scriptscriptstyle \textcolor{red}{\cdots}}
          & \multicolumn{1}{c}{\scriptscriptstyle \textcolor{red}{M\text{\!--\!}3}}
          & \multicolumn{1}{c}{\scriptscriptstyle \textcolor{red}{M\text{\!--\!}2}}
          & \multicolumn{1}{c}{\scriptscriptstyle \textcolor{red}{M\text{\!--\!}1}}
          & \multicolumn{1}{c}{\scriptscriptstyle \textcolor{red}{M}}
        \end{array}
      \right]
    }
  \end{align}
  \hrulefill
  \vspace*{4pt}
\end{figure*}
The coefficients $g_0, g_1,\ldots,g_L$ provide an alternative
parameterization of the generator set $\{z_1,\ldots,z_L\}$. Defining the polynomial
\begin{equation}\label{eq:Polynomial}
  P_{\boldsymbol{g}}(z)  \triangleq \sum_{\ell=0}^{L} g_\ell\, z^{\ell},
\end{equation}
with coefficient vector
${\boldsymbol g} = [g_0, \ldots, g_L]^{\mathsf T}$ and $g_0 \neq 0$, the
generators $z_1,\ldots,z_L$ are precisely the roots of
$P_{\boldsymbol{g}}(z)$
\begin{equation}\label{eq:PolynomialRelation}
  P_{\boldsymbol{g}}(z) = 0
  \quad\Longleftrightarrow\quad
  z \in \{z_1,\ldots,z_L\}.
\end{equation}
If all roots of $P_{\boldsymbol{g}}(z)$ lie on the unit circle, then reversing
and conjugating the coefficient sequence reproduces the same polynomial up to a
constant, i.e., $g_\ell = c\, g_{L-\ell}^*$ for $\ell = 0, 1, \ldots, L$ with
$|c| = 1$. Such polynomials are called self-inversive. The constant is an artefact of the arbitrary scaling of
${\boldsymbol g}$ and can always be removed: replacing ${\boldsymbol g}$ by
$e^{{\mathsf j}\gamma}{\boldsymbol g}$ with $\gamma = -\tfrac{1}{2}\arg c$ yields
\begin{equation}\label{eq:selfinversive}
  g_\ell = g_{L-\ell}^*, \qquad \ell = 0, 1, \ldots, L,
\end{equation}
while leaving the roots of $P_{\boldsymbol{g}}(z)$, the range of
${\mathcal{T}}(\boldsymbol{g})$, and all rank conditions considered in this
paper unchanged. Coefficient vectors satisfying~\eqref{eq:selfinversive} are
called centro-Hermitian, and this normalization is adopted throughout.  The converse does not hold: centro-Hermitian coefficients do not in general guarantee that the roots of $P_{\boldsymbol{g}}(z)$ lie on the unit circle,
which is why this has to be verified separately for every coefficient vector.

Since ${\boldsymbol A}_M(z_1,\ldots,z_L)$ is a Vandermonde matrix with
distinct generators, it has full column rank $L$ whenever $M \geq L$. For
the thinned array, however, this guarantee no longer holds: the
row-subsampled matrix
${\boldsymbol A}_{\mathcal{I}}(z_1,\ldots,z_L) =
{\boldsymbol T}_{\mathcal{I}}\,{\boldsymbol A}_M(z_1,\ldots,z_L)$
may become rank-deficient even when ${\boldsymbol A}_M(z_1,\ldots,z_L)$
itself has full rank. Characterizing the conditions under which such rank
deficiencies occur is the central objective of this paper. In essence, this
amounts to determining the \emph{spark} of the generalized Vandermonde
matrix family ${\boldsymbol A}_{\mathcal{I}}(z_1,\ldots,z_L)$, i.e., the
smallest number of columns that can be made linearly dependent by an
admissible choice of distinct generators $z_1,\ldots,z_L$ on the unit
circle.

Before stating the problem, we collect the standing assumptions that are used
throughout the paper.
\begin{itemize}[leftmargin=2.2em]
  \item[(A1)] The underlying sensor grid is critically sampled, $d = \lambda/2$, and all
    generators lie on the unit circle, $|z_i| = 1$ for $i = 1,\ldots,L$.
  \item[(A2)] The generators are pairwise \emph{distinct},
    $z_i \neq z_j$ for $i \neq j$ such that
    ${\boldsymbol A}_M(z_1,\ldots,z_L)$ is full rank.
  \item[(A3)] The first and the last sensor are always retained,
    $1, M \in \mathcal{I}$, so that the aperture of the thinned array equals that of
    the underlying ULA.
  \item[(A4)] At least one sensor is omitted, $M_{\mathcal{I}} < M$, i.e.,
    $\bar{\mathcal{I}} \neq \emptyset$.
  \item[(A5)] The number of sources satisfies $L < M$, so that
    ${\mathcal{T}}({\boldsymbol g})$ in~\eqref{def:ToeplitzEven}--\eqref{def:ToeplitzOdd}
    has at least one row. Note that $L < M_{\mathcal{I}}$ is \emph{not} assumed a
    priori; that it is in fact necessary for identifiability is a consequence of the
    theory, see Corollary~\ref{crl:corollary1}\,.
\end{itemize}
Note that (A1) further implies that $g_0 \neq 0$.
With these assumptions in place the central problem considered in this paper is stated as follows.
\begin{problem*}[Ambiguities in sparse arrays]\label{prob:main}
 Given a thinned ULA with index set $\mathcal{I}$ obtained from an $M$-element
  full ULA, and a number of sources $L$ satisfying (A1)--(A5), determine all
  generator sets $\{z_1,\ldots,z_L\}$ on the unit circle for which the steering
  matrix
  \begin{equation*}
    {\boldsymbol A}_{\mathcal{I}}(z_1,\ldots,z_L)
    = {\boldsymbol T}_{\mathcal{I}}\,{\boldsymbol A}_M(z_1,\ldots,z_L)
  \end{equation*}
  loses rank.  Equivalently, using the parameterization in terms of the centro-Hermitian
  polynomial coefficients, determine all coefficient vectors
  $\boldsymbol{g} = [g_0, g_1,\ldots,g_L]^{\mathsf T}$ whose associated
  polynomial $P_{\boldsymbol{g}}(z)$ has $L$ distinct roots on the unit circle
  and for which
  \begin{equation}\label{eq:CondRankDef}
    \operatorname{rank}\!\big({\boldsymbol A}_{\mathcal{I}}(z_1,\ldots,z_L)\big)
    < L,
  \end{equation}
  where $z_1,\ldots,z_L$ are these roots.
\end{problem*}
\begin{figure*}[!tb]
  \input{RankDeficiencyToeplitzFromSpaces_new_modernized.tex}
  \caption{A: For rank deficient ${\boldsymbol A}_{\mathcal{I}}$ there exists a vector ${\boldsymbol x} \in {\mathcal N}\big( {\boldsymbol A}_{\mathcal{I}} \big)$ such that ${\boldsymbol y}_{\mathcal I} = {\boldsymbol A}_{\mathcal{I}} {\boldsymbol x} = {\boldsymbol 0} $ cf. \eqref{eq:pc-a}. The product ${\boldsymbol A}_M {\boldsymbol x}$ yields a sparse hypothetical full array received signal vector ${\boldsymbol y}_M \in {\mathcal R} \big({\boldsymbol A}_M  \big)$ corresponding to ${\boldsymbol x}$. B: This augmented full array received vector has zeros at the retained sensor positions and generally nonzero auxiliary virtual measurements at the removed sensor locations, cf. \eqref{eq:pc-c}. C: The full array augmented measurement vector is per construction ${\boldsymbol y}_M \in {\mathcal R} \big({\boldsymbol A}_M  \big)$, such that the Toeplitz nullspace property ${\mathcal T}({\boldsymbol g}) {\boldsymbol y}_M = {\boldsymbol 0}$ holds, cf.  \eqref{eq:pc-d}, which due to the sparsity of ${\boldsymbol y}_M$ can further be collapsed to ${\mathcal T}_{\bar{\mathcal I}}({\boldsymbol g}) {\boldsymbol y}_{\bar{\mathcal I}} = {\boldsymbol 0}$ as illustrated in  D.
  } \label{fig:SchemeRankDeficiency}
\end{figure*}

\section{Algebraic Characterization of Ambiguity Sets}\label{sec:MainResults}
In this section we present the main results of the paper, namely the
algebraic characterization of ambiguity sets in sparse arrays for an
arbitrary number of sources. To address the ambiguity problem in sparse
arrays, we equivalently seek conditions under which the null space of
${\boldsymbol A}_{\mathcal{I}}(z_1,\ldots,z_L)$ is nontrivial, that is,
\begin{equation}\label{eq:CondRankDefA_I}
  \dim\!\big({\cal N}({\boldsymbol A}_{\mathcal{I}}(z_1,\ldots,z_L))\big)  \triangleq D
  \geq 1,
\end{equation}
or, equivalently,
$\operatorname{rank}({\boldsymbol A}_{\mathcal{I}}(z_1,\ldots,z_L)) = L - D$,
where the integer $D$ quantifies the degree of rank deficiency and equals the
number of linearly independent source amplitude vectors that the thinned array
cannot distinguish from the all-zero scenario.

As we shall see, this rank deficiency is characterized most conveniently
through the removed sensors rather than the retained ones. We therefore
introduce the complementary index set
\begin{equation}\label{def:DefComplementaryIndexSet}
  \bar{\mathcal{I}}  \triangleq \{1, 2, \ldots, M\} \setminus \mathcal{I}
  = \{\bar{m}_1, \ldots, \bar{m}_{M_{\bar{\mathcal{I}}}}\},
\end{equation}
with $M_{\bar{\mathcal{I}}} = M - M_{\mathcal{I}}$ removed sensors, and the
corresponding selection matrix
${\boldsymbol T}_{\bar{\mathcal{I}}} \in
\mathbb{B}^{M_{\bar{\mathcal{I}}} \times M}$ defined according
to~\eqref{def:DefSelectionMatrix}. Right-multiplication by ${\boldsymbol T}_{\bar{\mathcal{I}}}^{\mathsf{T}}$
selects the columns indexed by the removed positions. Furthermore, the two
selection matrices satisfy
${\boldsymbol T}_{\mathcal{I}}^{\mathsf{T}}{\boldsymbol T}_{\mathcal{I}}
+ {\boldsymbol T}_{\bar{\mathcal{I}}}^{\mathsf{T}}{\boldsymbol T}_{\bar{\mathcal{I}}}
= {\boldsymbol I}_M$, so that a vector annihilated by
${\boldsymbol T}_{\mathcal{I}}$ equals
${\boldsymbol T}_{\bar{\mathcal{I}}}^{\mathsf{T}}{\boldsymbol y}_{\bar{\mathcal{I}}}$
and is determined by its entries at the removed positions alone.

With this notation, the central result of the paper states that the rank
deficiency of the thinned array manifold, a matrix of size
$M_{\mathcal{I}} \times L$ whose entries depend nonlinearly on the generators,
is exactly mirrored by the rank deficiency of a complementary Toeplitz
submatrix, a matrix of size $(M - L) \times M_{\bar{\mathcal{I}}}$ whose
entries depend \emph{linearly} on the coefficient vector ${\boldsymbol g}$.

\begin{theorem}[Rank Condition for Generalized Vandermonde Matrices]
  \label{thm:rank_condition}
  Let
  ${\boldsymbol A}_{\mathcal{I}}(z_1,\ldots,z_L)
  \in \mathbb{C}^{M_{\mathcal{I}}\times L}$
  be a generalized Vandermonde matrix as defined in~\eqref{def:DefAtilde},
  with index set $\mathcal{I}$ as in~\eqref{def:DefIndexSet}, and let
  $L \leq M_{\mathcal{I}}$.
  Suppose that assumptions (A1)--(A5) hold, and let
  ${\boldsymbol g}=[g_0,\ldots,g_L]^{\mathsf T}$
  denote the coefficient vector of the centro-Hermitian polynomial
  $P_{\boldsymbol g}(z)$ in~\eqref{eq:Polynomial}, whose roots are the
  \emph{distinct} generators $z_1,\ldots,z_L$.
  Define the complementary submatrix
  \begin{equation}\label{def:DefG_0}
    {\mathcal T}_{\bar{\mathcal I}}({\boldsymbol g})
    \triangleq
    {\mathcal T}({\boldsymbol g})
    {\boldsymbol T}_{\bar{\mathcal I}}^{\mathsf T}
    \in
    \mathbb C^{(M-L)\times(M-M_{\mathcal I})}.
  \end{equation}
  Then
  ${\boldsymbol A}_{\mathcal I}(z_1,\ldots,z_L)$
  is rank deficient if and only if
  ${\mathcal T}_{\bar{\mathcal I}}({\boldsymbol g})$
  is rank deficient.

  Moreover, the dimensions of the nullspaces coincide, i.e.,
  \begin{equation}\label{eq:thm-rank-exact}
    \dim \mathcal N\!\big(
      {\boldsymbol A}_{\mathcal I}(z_1,\ldots,z_L)
    \big)
    =
    \dim \mathcal N\!\big(
      {\mathcal T}_{\bar{\mathcal I}}({\boldsymbol g})
    \big)
    \triangleq D,
  \end{equation}
  where, in the rank-deficient case,
  \begin{equation}
    1 \leq D
    \leq
    \min\{L,\,M-M_{\mathcal I}\}.
  \end{equation}
\end{theorem}
\makeatletter
\makeatother
\begin{IEEEproof}
By~\eqref{eq:CondRankDefA_I}, the rank deficits of
  ${\boldsymbol A}_{\mathcal{I}}$ and
  ${\mathcal{T}}_{\bar{\mathcal{I}}}({\boldsymbol g})$
  are the dimensions of their null spaces. These are subspaces of
  $\mathbb{C}^{L}$ and $\mathbb{C}^{M_{\bar{\mathcal{I}}}}$ and cannot be
  compared directly, so the proof constructs a linear bijective map between
  them, which then implies equality of their dimensions. 

  By~\eqref{def:DefAtilde}, ${\boldsymbol A}_{\mathcal{I}} =
  {\boldsymbol T}_{\mathcal{I}}{\boldsymbol A}_M$, so a vector ${\boldsymbol x}$ in
  ${\cal N}({\boldsymbol A}_{\mathcal{I}})$ produces a response of the \emph{full}
  ULA, ${\boldsymbol y}_M = {\boldsymbol A}_M{\boldsymbol x}$, that vanishes at every
  retained sensor. Writing ${\boldsymbol y}_{\bar{\mathcal I}} =
  {\boldsymbol T}_{\bar{\mathcal{I}}}{\boldsymbol y}_M$ for its entries at the
  missing sensor positions, we obtain the chain of equivalences for rank deficient ${\boldsymbol A}_{\mathcal{I}}$
  \begin{subequations}\label{eq:proof-chain}
    \begin{align}
      &{\boldsymbol A}_{\mathcal{I}}{\boldsymbol x} = {\boldsymbol 0}
      \label{eq:pc-a}\\
      \Longleftrightarrow\;\;
      &{\boldsymbol T}_{\mathcal{I}}\,{\boldsymbol A}_M{\boldsymbol x}
      = {\boldsymbol T}_{\mathcal{I}}\,{\boldsymbol y}_M = {\boldsymbol 0}
      \label{eq:pc-b}\\
      \Longleftrightarrow\;\;
      &{\boldsymbol y}_M \in {\cal V}
      \triangleq
      {\cal R}({\boldsymbol A}_M)
      \cap
      {\cal N}({\boldsymbol T}_{\mathcal{I}})
      \label{eq:pc-c}\\
      \Longleftrightarrow\;\;
      &{\mathcal{T}}({\boldsymbol g})\,{\boldsymbol y}_M = {\boldsymbol 0}
      \;\;\text{and}\;\;
      {\boldsymbol T}_{\mathcal{I}}{\boldsymbol y}_M = {\boldsymbol 0}
      \label{eq:pc-d}\\
      \Longleftrightarrow\;\;
      &{\mathcal{T}}({\boldsymbol g})\,
      {\boldsymbol T}_{\bar{\mathcal{I}}}^{\mathsf{T}}
      {\boldsymbol y}_{\bar{\mathcal I}}
      =
      {\mathcal{T}}_{\bar{\mathcal{I}}}({\boldsymbol g})\,
      {\boldsymbol y}_{\bar{\mathcal I}}
      = {\boldsymbol 0}.
      \label{eq:pc-e}
    \end{align}
  \end{subequations}

  Here \eqref{eq:pc-b} follows from~\eqref{def:DefAtilde}, and~\eqref{eq:pc-c} restates it,
  since ${\boldsymbol y}_M = {\boldsymbol A}_M{\boldsymbol x}$ lies in
  ${\cal R}({\boldsymbol A}_M)$ by construction. Step~\eqref{eq:pc-d} uses that
  under (A2) the generators are distinct, so that
  ${\cal R}\big({\boldsymbol A}_M\big) = {\cal N}\big({\mathcal{T}}({\boldsymbol g})\big)$
  by~\eqref{eq:subspace_property_equality}, where $\{z_1,\ldots,z_L\}$ and ${\boldsymbol g}$ are related through~\eqref{eq:PolynomialRelation}. Finally, a vector with
  ${\boldsymbol T}_{\mathcal{I}}{\boldsymbol y}_M = {\boldsymbol 0}$ is zero outside
  $\bar{\mathcal{I}}$ and hence satisfies ${\boldsymbol y}_M =
  {\boldsymbol T}_{\bar{\mathcal{I}}}^{\mathsf{T}}{\boldsymbol y}_{\bar{\mathcal I}}$,
  which together with~\eqref{def:DefG_0} gives~\eqref{eq:pc-e}.

  Thus, rank deficiency of ${\boldsymbol A}_{\mathcal I}$ is equivalent to the
existence of a nonzero set of virtual measurements
${\boldsymbol y}_{\bar{\mathcal I}}$ at the removed sensor positions that
completes the zero response ${\boldsymbol y}_{\mathcal I}$ at the retained
sensors to a valid full-ULA response
${\boldsymbol y}_M \in {\cal R}({\boldsymbol A}_M)$.

  The chain of equivalences in~\eqref{eq:proof-chain} is illustrated in Fig.~\ref{fig:SchemeRankDeficiency}\,. It shows that ${\boldsymbol A}_{\mathcal{I}}(z_1, \ldots, z_L)$ is rank deficient if and only if ${\mathcal{T}}_{\bar{\mathcal{I}}}({\boldsymbol g})$ is rank deficient. To obtain the
  stronger statement~\eqref{eq:thm-rank-exact}, consider the two maps
  ${\boldsymbol x} \mapsto {\boldsymbol y}_M = {\boldsymbol A}_M{\boldsymbol x}$ and
  ${\boldsymbol y}_M \mapsto {\boldsymbol y}_{\bar{\mathcal I}} =
  {\boldsymbol T}_{\bar{\mathcal{I}}}{\boldsymbol y}_M$ appearing
  in~\eqref{eq:proof-chain}. Both are linear. The first is injective because
  ${\boldsymbol A}_M$ has full column rank under (A2), and by~\eqref{eq:pc-c} it
  maps ${\cal N}({\boldsymbol A}_{\mathcal{I}})$ onto ${\cal V}$. The second is
  injective on ${\cal V}$, since ${\boldsymbol y}_M$ is recovered from
  ${\boldsymbol y}_{\bar{\mathcal I}}$ as
  ${\boldsymbol y}_M = {\boldsymbol T}_{\bar{\mathcal{I}}}^{\mathsf{T}}{\boldsymbol y}_{\bar{\mathcal I}}$,
  and by~\eqref{eq:pc-e} it maps ${\cal V}$ onto
  ${\cal N}({\mathcal{T}}_{\bar{\mathcal{I}}}({\boldsymbol g}))$. Their
  composition ${\boldsymbol x} \mapsto {\boldsymbol y}_{\bar{\mathcal I}} =
  {\boldsymbol T}_{\bar{\mathcal{I}}}{\boldsymbol A}_M{\boldsymbol x}$ is therefore
  a linear bijective map from ${\cal N}({\boldsymbol A}_{\mathcal{I}})$ onto
  ${\cal N}({\mathcal{T}}_{\bar{\mathcal{I}}}({\boldsymbol g}))$, which carries a
  basis of the one null space to a basis of the other, so the two dimensions
  agree. Since the dimension of a null space equals the number of columns minus
  the rank, this is~\eqref{eq:thm-rank-exact}, and the bounds
  $1 \leq D \leq \min\{L,\, M - M_{\mathcal{I}}\}$ hold because a null space
  cannot have larger dimension than the number of columns of either matrix.
\end{IEEEproof}

By~\eqref{eq:thm-rank-exact}, Theorem~\ref{thm:rank_condition} equates not
merely the two rank conditions but the two rank deficits, so that the
complementary submatrix determines the degree of ambiguity and not only its
presence. It is moreover usable in the constructive direction: since
${\mathcal{T}}_{\bar{\mathcal{I}}}({\boldsymbol g})$ depends linearly on the
$\lfloor L/2 \rfloor + 1$ free coefficients of ${\boldsymbol g}$, one may search
for coefficient vectors at which it loses rank and obtain the ambiguity sets as
the roots of $P_{\boldsymbol{g}}(z)$, each with its deficit $D$. All examples in
Section~\ref{sec:Examples} proceed in this direction.

\begin{remark}[Generators off the unit circle]
  Theorem~\ref{thm:rank_condition} does not require the generators to lie on
  the unit circle. It holds for generalized Vandermonde matrices
  ${\boldsymbol A}_{\mathcal{I}}(z_1,\ldots,z_L)$ with distinct generators of
  arbitrary nonzero magnitude. In this more general setting, the Toeplitz
  matrix ${\mathcal{T}}({\boldsymbol g})$ and its complementary submatrix
  ${\mathcal{T}}_{\bar{\mathcal{I}}}({\boldsymbol g})$ are built from a
  coefficient vector that does not necessarily satisfy the
  centro-Hermitian property $g_{\ell} = g_{L-\ell}^*$. The proof carries
  over verbatim provided  that all generators are nonzero,
  so that ${\boldsymbol A}_M$ is Vandermonde with $g_0 = \prod_\ell(-z_\ell) \neq 0$
  and~\eqref{eq:subspace_property_equality} holds with
  ${\mathcal{T}}({\boldsymbol g})$ of full row rank $M - L$.
\end{remark}

\begin{remark}[Undersampled grids]\label{rmk:undersampled}
  Assumption (A1) fixes $d = \lambda/2$, so that the generators
  $z_i = e^{{\mathsf j}\theta_i}$ range over the entire unit circle. For
  $d < \lambda/2$ they are confined to the arc $|\theta| \leq 2\pi d/\lambda$.
  Theorem~\ref{thm:rank_condition} itself is unaffected, being a purely algebraic
  statement, but a rank drop then certifies a physical ambiguity only if all $L$
  associated generators fall inside that arc.
\end{remark}

\begin{remark}[Repeated generators]\label{rmk:repeated}
  Assumption (A2) cannot be dropped. If two generators coincide, then
  ${\boldsymbol A}_{\mathcal{I}}(z_1,\ldots,z_L)$ is trivially rank deficient
  (two identical columns), whereas
  ${\mathcal{T}}_{\bar{\mathcal{I}}}({\boldsymbol g})$ may retain full column
  rank, so that~\eqref{eq:thm-rank-exact} fails. For example, take $M = 9$,
  $L = 3$, $\bar{\mathcal{I}} = \{2,3,4,5,6\}$ and the generator multiset
  $\{1,\, 1,\, e^{{\mathsf j}1.1}\}$: the thinned manifold has rank $2 < 3$,
  while the corresponding $6 \times 5$ complementary submatrix has full column
  rank $5$.

  It is worth noting which part of the argument breaks down. The annihilation
  property~\eqref{eq:subspace_property} continues to hold, since it depends only
  on the roots of $P_{\boldsymbol g}(z)$ and not on their multiplicities. What
  fails is the equality~\eqref{eq:subspace_property_equality}: for repeated roots
  $\operatorname{rank}({\boldsymbol A}_M) < L
  = \dim {\cal N}({\mathcal{T}}({\boldsymbol g}))$, so that
  ${\mathcal R}({\boldsymbol A}_M) \subsetneq
  {\cal N}({\mathcal{T}}({\boldsymbol g}))$; in the above example the two
  dimensions are $2$ and $3$. Consequently the dimension count, and with it both
  the converse implication and the equality of the two null space dimensions
  in~\eqref{eq:thm-rank-exact}, is no longer valid.

  Since ambiguity sets consist of $L$ distinct generators by definition, this
  restriction is without practical consequence. It does, however, have to be
  respected in constructions; see the choice of $\alpha$ in the footnote of the
  proof of Corollary~\ref{crl:corollary4b}\,.
\end{remark}
\begin{remark}[Computational advantage]\label{rmk:computational}
  Theorem~\ref{thm:rank_condition} replaces a rank test on the
  $M_{\mathcal{I}} \times L$ manifold, which has to be repeated over a continuum
  of generator sets, by a rank test on the
  $(M - L) \times (M - M_{\mathcal{I}})$ matrix
  ${\mathcal{T}}_{\bar{\mathcal{I}}}({\boldsymbol g})$, which depends
  \emph{linearly} on the $\lfloor L/2 \rfloor + 1$ free coefficients of
  ${\boldsymbol g}$. The gain is largest in the two regimes that matter most in
  practice. When the number of sources approaches the aperture, $M - L$ is small
  and ${\mathcal{T}}_{\bar{\mathcal{I}}}({\boldsymbol g})$ has very few rows; for
  $M - L = 1$ it degenerates to a single row, and for $M - L = 2$ to a
  $2 \times (M - M_{\mathcal{I}})$ matrix whose rank drop is characterized by
  $2 \times 2$ minors. Conversely, when only a few sensors are removed,
  $M - M_{\mathcal{I}}$ is small and the criterion reduces to a rank condition on a
  matrix with very few columns; for a single omitted sensor it reduces to the
  scalar condition $g_{i-1} = 0$, cf.\ Example~\ref{exmpl:1}\,. This structural
  reduction is used implicitly in all examples of
  Section~\ref{sec:Examples}\,.
\end{remark}

\subsection{Rotation invariance}
The rank of a generalized Vandermonde matrix
${\boldsymbol A}_{\mathcal{I}}(z_1, \ldots, z_L)$ is invariant under a
common rotation of its generators. That is, if
$\{\grave{z}_1, \ldots, \grave{z}_L\}$ is a generator set at which the
matrix drops rank, then so is $\{\acute{z}_1, \ldots, \acute{z}_L\}$ with
\begin{equation}\label{eq:rotation_of_z}
  \acute{z}_\ell  \triangleq \grave{z}_\ell\, e^{{\mathsf j}\alpha},
  \qquad \ell = 1, \ldots, L,
\end{equation}
for any $\alpha \in (-\pi, \pi]$, and
\begin{equation}
  \operatorname{rank}\!\big({\boldsymbol A}_{\mathcal{I}}(\acute{z}_1, \ldots, \acute{z}_L)\big)
  = \operatorname{rank}\!\big({\boldsymbol A}_{\mathcal{I}}(\grave{z}_1, \ldots, \grave{z}_L)\big).
\end{equation}
This follows from the observation that a common rotation of the generators
of a full Vandermonde matrix amounts to left multiplication by the
nonsingular diagonal matrix
\begin{equation}
  {\boldsymbol D}_{2,\alpha}  \triangleq \operatorname{diag}\!\big(1,\, e^{{\mathsf j}\alpha},\, \ldots,\, e^{{\mathsf j}(M-1)\alpha}\big),
\end{equation}
that is,
$
  {\boldsymbol A}_M(\acute{z}_1, \ldots, \acute{z}_L)
  = {\boldsymbol D}_{2,\alpha}\,
  {\boldsymbol A}_M(\grave{z}_1, \ldots, \grave{z}_L)$.\\
For the thinned array, note that selecting rows from a diagonal matrix
commutes with the selection, i.e.,
\begin{equation}\label{eq:commute_T_D}
  {\boldsymbol T}_{\mathcal{I}}\,{\boldsymbol D}_{2,\alpha}
  = {\boldsymbol D}_{\mathcal{I},2,\alpha}\,{\boldsymbol T}_{\mathcal{I}},
\end{equation}
where
\begin{align}\label{def:DIalpha}
  {\boldsymbol D}_{\mathcal{I},2,\alpha}
 & \triangleq {\boldsymbol T}_{\mathcal{I}}\,{\boldsymbol D}_{2,\alpha}\,
  {\boldsymbol T}_{\mathcal{I}}^{\mathsf{T}}\\
 & = \operatorname{diag}\!\big(e^{{\mathsf j}(m_1 - 1)\alpha},\, \ldots,\, e^{{\mathsf j}(m_{M_{\mathcal{I}}} - 1)\alpha}\big)\nonumber 
\end{align}
is nonsingular. Hence
\begin{align}
  {\boldsymbol A}_{\mathcal{I}}(\acute{z}_1, \ldots, \acute{z}_L)
  &= {\boldsymbol T}_{\mathcal{I}}\,
  {\boldsymbol A}_M(\acute{z}_1, \ldots, \acute{z}_L) \notag\\
  &= {\boldsymbol T}_{\mathcal{I}}\,{\boldsymbol D}_{2,\alpha}\,
  {\boldsymbol A}_M(\grave{z}_1, \ldots, \grave{z}_L) \notag\\
  &= {\boldsymbol D}_{\mathcal{I},2,\alpha}\,
  {\boldsymbol T}_{\mathcal{I}}\,
  {\boldsymbol A}_M(\grave{z}_1, \ldots, \grave{z}_L) \notag\\
  &= {\boldsymbol D}_{\mathcal{I},2,\alpha}\,
  {\boldsymbol A}_{\mathcal{I}}(\grave{z}_1, \ldots, \grave{z}_L).
\end{align} 
The rotation invariance carries over to the complementary
submatrix ${\mathcal{T}}_{\bar{\mathcal{I}}}({\boldsymbol g})$, as stated in
the following corollary.
\begin{corollary}[Rotation invariance]\label{crl:rotation_invariance}
  Let $\grave{\boldsymbol g}$ and $\acute{\boldsymbol g}$ be two coefficient
  vectors whose entries are related through
  \begin{equation}\label{eq:rotationPolCoef}
    \acute{g}_\ell = e^{{\mathsf j}\beta}\, e^{-{\mathsf j}\ell\alpha}\,
    \grave{g}_\ell,
    \qquad \ell = 0, 1, \ldots, L,
  \end{equation}
  for $\alpha \in (-\pi, \pi]$ and
  \begin{equation}\label{eq:def_beta}
    \beta  \triangleq
    \begin{cases}
      K\alpha, & \text{for even $L = 2K$,}\\[2pt]
      \big(K + \tfrac{1}{2}\big)\alpha, & \text{for odd $L = 2K + 1$.}
    \end{cases}
  \end{equation}
  Then the roots of $P_{\acute{\boldsymbol g}}(z)$ are obtained from those of
  $P_{\grave{\boldsymbol g}}(z)$ by the common rotation~\eqref{eq:rotation_of_z},
  and ${\mathcal{T}}_{\bar{\mathcal{I}}}(\acute{\boldsymbol g})$ is rank
  deficient if and only if
  ${\mathcal{T}}_{\bar{\mathcal{I}}}(\grave{\boldsymbol g})$ is rank deficient.
\end{corollary}
\begin{IEEEproof} The Toeplitz matrices defined in~\eqref{def:ToeplitzEven}
  and~\eqref{def:ToeplitzOdd} for even and odd source numbers $L$,
  respectively, satisfy
  \begin{equation}\label{eq:TacuteTbreve}
    {\mathcal{T}}(\acute{\boldsymbol g})
    = e^{{\mathsf j}\beta}\,
    {\boldsymbol D}_{1,\alpha}\,
    {\mathcal{T}}(\grave{\boldsymbol g})\,
    {\boldsymbol D}_{2,\alpha}^*,
  \end{equation}
  where
  \begin{align}
    {\boldsymbol D}_{1,\alpha}
    & \triangleq\operatorname{diag}\!\big(1,\, e^{{\mathsf j}\alpha},\, \ldots,\, e^{{\mathsf j}(M - L - 1)\alpha}\big),
  \end{align}
  and $\beta$ is given by~\eqref{eq:def_beta}.
  Using the complementary selection matrix ${\boldsymbol T}_{\bar{\mathcal{I}}}$
  defined according to~\eqref{def:DefSelectionMatrix} for the index set
  $\bar{\mathcal{I}}$ in~\eqref{def:DefComplementaryIndexSet}, we have the
  identity
  \begin{equation}\label{eq:identityDiagT}
    {\boldsymbol D}_{2,\alpha}^*\,{\boldsymbol T}_{\bar{\mathcal{I}}}^{\mathsf{T}}
    = {\boldsymbol T}_{\bar{\mathcal{I}}}^{\mathsf{T}}\,
    {\boldsymbol D}_{\bar{\mathcal{I}},2,\alpha}^*,
  \end{equation}
  where
  \begin{equation}
    {\boldsymbol D}_{\bar{\mathcal{I}},2,\alpha}
    = {\boldsymbol T}_{\bar{\mathcal{I}}}\,{\boldsymbol D}_{2,\alpha}\,
    {\boldsymbol T}_{\bar{\mathcal{I}}}^{\mathsf{T}}
  \end{equation}
  is nonsingular and diagonal. Combining~\eqref{eq:TacuteTbreve},
  \eqref{eq:identityDiagT}, and definition~\eqref{def:DefG_0} yields
  \begin{align}
    {\mathcal{T}}_{\bar{\mathcal{I}}}(\acute{\boldsymbol g})
    &= {\mathcal{T}}(\acute{\boldsymbol g})\,
    {\boldsymbol T}_{\bar{\mathcal{I}}}^{\mathsf{T}} \notag\\
    &= e^{{\mathsf j}\beta}\,
    {\boldsymbol D}_{1,\alpha}\,
    {\mathcal{T}}(\grave{\boldsymbol g})\,
    {\boldsymbol D}_{2,\alpha}^*\,
    {\boldsymbol T}_{\bar{\mathcal{I}}}^{\mathsf{T}} \notag\\
    &= e^{{\mathsf j}\beta}\,
    {\boldsymbol D}_{1,\alpha}\,
    {\mathcal{T}}(\grave{\boldsymbol g})\,
    {\boldsymbol T}_{\bar{\mathcal{I}}}^{\mathsf{T}}\,
    {\boldsymbol D}_{\bar{\mathcal{I}},2,\alpha}^* \notag\\
    &= e^{{\mathsf j}\beta}\,
    {\boldsymbol D}_{1,\alpha}\,
    {\mathcal{T}}_{\bar{\mathcal{I}}}(\grave{\boldsymbol g})\,
    {\boldsymbol D}_{\bar{\mathcal{I}},2,\alpha}^*.
  \end{align}
  Since ${\boldsymbol D}_{1,\alpha}$ and
  ${\boldsymbol D}_{\bar{\mathcal{I}},2,\alpha}^*$ are nonsingular diagonal
  matrices and $e^{{\mathsf j}\beta}$ is a unit-modulus scalar, the rank is
  preserved:
  \begin{equation}
    \operatorname{rank}\!\big({\mathcal{T}}_{\bar{\mathcal{I}}}(\acute{\boldsymbol g})\big)
    = \operatorname{rank}\!\big({\mathcal{T}}_{\bar{\mathcal{I}}}(\grave{\boldsymbol g})\big).
  \end{equation}
\end{IEEEproof}
The rank of ${\boldsymbol A}_{\mathcal{I}}(z_1, \ldots, z_L)$ is invariant
under a common rotation of its generators, so every ambiguity set comes with a
continuum of rotated copies. A normalization constraint is therefore typically
imposed~\cite{matterAmbiguitiesDoAEstimation2022} to ensure uniqueness of the
solutions obtained from Theorem~\ref{thm:rank_condition} when characterizing the
ambiguity sets of a given array geometry. A simple
choice is to fix one root to an arbitrary location on the unit circle,
for instance, $z_1 = 1$, which is equivalent to requiring the polynomial
coefficients to sum to zero, $P_{\boldsymbol g}(1) = {\boldsymbol g}^{\mathsf{T}}{\boldsymbol 1} = 0$.
This constraint resolves the rotation ambiguity up to a permutation of the
roots, i.e., up to the choice of which root is fixed to one. An
alternative is to fix the phase of the constant coefficient $g_0$ to zero,
i.e., $\operatorname{Im}(g_0) = 0$ and $\operatorname{Re}(g_0) > 0$.
\subsection{Conjugate invariance}

In addition to the rotation invariance, the generalized Vandermonde matrix
${\boldsymbol A}_{\mathcal{I}}(z_1,\ldots,z_L)$ and its complementary
submatrix ${\mathcal{T}}_{\bar{\mathcal{I}}}({\boldsymbol g})$ exhibit a
conjugate invariance, as stated next.

\begin{corollary}[Conjugate invariance]\label{crl:conjugate_invariance}
  The complementary submatrix
  ${\mathcal{T}}_{\bar{\mathcal{I}}}({\boldsymbol g})$ is rank deficient if
  and only if ${\mathcal{T}}_{\bar{\mathcal{I}}}({\boldsymbol g}^*)$ is rank
  deficient.
\end{corollary}

\begin{IEEEproof}
  Inspection of~\eqref{def:ToeplitzEven} and~\eqref{def:ToeplitzOdd} shows
  that replacing each coefficient $g_\ell$ by $g_\ell^*$ replaces every entry
  of ${\mathcal{T}}({\boldsymbol g})$ by its complex conjugate. Hence
  ${\mathcal{T}}({\boldsymbol g}^*) = {\mathcal{T}}({\boldsymbol g})^*$, and
  since ${\boldsymbol T}_{\bar{\mathcal{I}}}$ is real,
  \begin{equation}
    {\mathcal{T}}_{\bar{\mathcal{I}}}({\boldsymbol g}^*)
    = {\mathcal{T}}({\boldsymbol g}^*)\,{\boldsymbol T}_{\bar{\mathcal{I}}}^{\mathsf{T}}
    = \big({\mathcal{T}}({\boldsymbol g})\,{\boldsymbol T}_{\bar{\mathcal{I}}}^{\mathsf{T}}\big)^*
    = {\mathcal{T}}_{\bar{\mathcal{I}}}({\boldsymbol g})^*.
  \end{equation}
  Since the rank of a matrix is invariant under complex conjugation,
  $\operatorname{rank}\!\big({\mathcal{T}}_{\bar{\mathcal{I}}}({\boldsymbol g}^*)\big)
  = \operatorname{rank}\!\big({\mathcal{T}}_{\bar{\mathcal{I}}}({\boldsymbol g})\big)$,
  which completes the proof.
\end{IEEEproof}

Although Corollary~\ref{crl:conjugate_invariance} is elementary, it
captures a useful symmetry: the generalized Vandermonde matrix
${\boldsymbol A}_{\mathcal{I}}(z_1,\ldots,z_L)$ loses rank for the
generator set $\{z_1,\ldots,z_L\}$ if and only if it loses rank for the
conjugate generator set $\{z_1^*,\ldots,z_L^*\}$. Indeed, if
$\{z_1,\ldots,z_L\}$ are the roots of $P_{\boldsymbol g}(z)$, then
$\{z_1^*,\ldots,z_L^*\}$ are the roots of $P_{{\boldsymbol g}^*}(z)$, since
$P_{{\boldsymbol g}^*}(z) = \big(P_{\boldsymbol g}(z^*)\big)^*$.
\subsection{Direct implications of Theorem~\ref{thm:rank_condition}}
In this subsection we discuss several implications that follow directly
from Theorem~\ref{thm:rank_condition}\,. The results show that the theorem is in agreement with established array processing theory.

The first implication confirms the common intuition that thinning a
$\lambda/2$-spaced ULA always introduces ambiguities in the resulting steering matrix.
\begin{corollary}[Violation of the spatial sampling theorem]\label{crl:corollary1}
  For a thinned ULA with $M_{\mathcal{I}} < M$ sensors and a number of
  sources $L \geq M_{\mathcal{I}}$, the thinned array steering matrix
  ${\boldsymbol A}_{\mathcal{I}}(z_1, \ldots, z_L)$ is always ambiguous:
  there exists a set of distinct generators $\{z_1, \ldots, z_L\}$ on the
  unit circle for which
  ${\boldsymbol A}_{\mathcal{I}}(z_1, \ldots, z_L)$ becomes rank deficient.
\end{corollary}

\begin{IEEEproof}
  It suffices to prove the corollary for $L = M_{\mathcal{I}}$, since rank
  deficiency at $L = M_{\mathcal{I}}$ implies rank deficiency for any
  $L > M_{\mathcal{I}}$ by selecting an arbitrary $M_{\mathcal{I}}$-subset
  of the columns. For $L = M_{\mathcal{I}}$, the complementary submatrix
  $\mathcal{T}_{\bar{\mathcal{I}}}(\boldsymbol{g}) \in
  \mathbb{C}^{(M - L) \times (M - M_{\mathcal{I}})}$ is square of size
  $M_{\bar{\mathcal{I}}} \times M_{\bar{\mathcal{I}}}$ with
  $M_{\bar{\mathcal{I}}} = M - M_{\mathcal{I}}$.

  Consider the centro-Hermitian polynomial
  \begin{equation}\label{eq:conv_b_h_g}
    P_{\boldsymbol{b}}(z)  \triangleq 1 + z^{M - 1}
    = \prod_{m = 1}^{M - 1}\!\big(z - e^{{\mathsf j}\pi(2m - 1)/(M - 1)}\big),
  \end{equation}
  whose $M - 1$ roots lie on the unit circle as the $(M-1)$-th roots of
  $-1$. Factorizing $P_{\boldsymbol{b}}(z)$ into two centro-Hermitian factors, yields
  \begin{equation}\label{eq:b_h_g_factorization}
    P_{\boldsymbol{b}}(z) = P_{\boldsymbol{h}}(z)\, P_{\boldsymbol{g}}(z),
  \end{equation}
  where
$P_{\boldsymbol{h}}(z)  \triangleq \sum_{\ell = 0}^{M_{\bar{\mathcal{I}}} - 1} h_\ell\, z^\ell$,
 $P_{\boldsymbol{g}}(z)  \triangleq \sum_{\ell = 0}^{L} g_\ell\, z^\ell$,
  with coefficient vectors $\boldsymbol{b} = \boldsymbol{e}_1 +
  \boldsymbol{e}_M \in \mathbb{B}^M$, for $\boldsymbol{e}_i$ denoting the
  $i$-th unit vector in $\mathbb{R}^M$,
  $\boldsymbol{h} \in \mathbb{C}^{M_{\bar{\mathcal{I}}}}$ satisfying the
  centro-Hermitian property $h_\ell = h_{M_{\bar{\mathcal{I}}} - 1 - \ell}^*$
  with $|h_0| = 1$, and $\boldsymbol{g} \in \mathbb{C}^{L + 1}$ satisfying
  $g_\ell = g_{L - \ell}^*$ with $|g_0| = 1$.

  The polynomial identity~\eqref{eq:b_h_g_factorization} can be expressed
  in matrix form as
  \begin{equation}\label{eq:conv_b_h_g_matrix}
    \boldsymbol{b}^{\mathsf{T}}
    = \boldsymbol{h}^{\mathsf{T}}\,\mathcal{T}(\boldsymbol{g}),
  \end{equation}
  where $\mathcal{T}(\boldsymbol{g}) \in
  \mathbb{C}^{(M - L) \times M}$ is the Toeplitz matrix defined
  in~\eqref{def:ToeplitzEven}--\eqref{def:ToeplitzOdd}. Right-multiplying
  both sides of~\eqref{eq:conv_b_h_g_matrix} by the complementary
  selection matrix $\boldsymbol{T}_{\bar{\mathcal{I}}}^{\mathsf{T}}$
  yields
  \begin{equation}\label{eq:conv_b_h_g_matrix_square}
    \boldsymbol{b}^{\mathsf{T}}\,\boldsymbol{T}_{\bar{\mathcal{I}}}^{\mathsf{T}}
    = \boldsymbol{h}^{\mathsf{T}}\,\mathcal{T}(\boldsymbol{g})\,
    \boldsymbol{T}_{\bar{\mathcal{I}}}^{\mathsf{T}}
    = \boldsymbol{h}^{\mathsf{T}}\,
    \mathcal{T}_{\bar{\mathcal{I}}}(\boldsymbol{g}).
  \end{equation}
  Since $\boldsymbol{b}$ has nonzero entries only at positions $1$ and
  $M$, both of which lie in $\mathcal{I}$ (the first and last sensors are
  always retained to preserve the aperture), the selection of columns
  indexed by $\bar{\mathcal{I}}$ from $\boldsymbol{b}^{\mathsf{T}}$
  yields
  $\boldsymbol{b}^{\mathsf{T}}\,\boldsymbol{T}_{\bar{\mathcal{I}}}^{\mathsf{T}}
  = \boldsymbol{0}^{\mathsf{T}}$, the zero row vector of length
  $M_{\bar{\mathcal{I}}}$. Combined with $\boldsymbol{h} \neq \boldsymbol{0}$
  (since $|h_0| = 1$),
  identity~\eqref{eq:conv_b_h_g_matrix_square} exhibits $\boldsymbol{h}$
  as a nontrivial left null vector of
  $\mathcal{T}_{\bar{\mathcal{I}}}(\boldsymbol{g})$. As
  $\mathcal{T}_{\bar{\mathcal{I}}}(\boldsymbol{g})$ is square, this
  forces it to be rank deficient.

  By Theorem~\ref{thm:rank_condition}\,, the corresponding generalized
  Vandermonde matrix
  ${\boldsymbol A}_{\mathcal{I}}(\grave{z}_1, \ldots, \grave{z}_L)$ is
  rank deficient, where $\{\grave{z}_1, \ldots, \grave{z}_L\}$ is the
  subset of roots of $P_{\boldsymbol{b}}(z)$ associated with the factor
  $P_{\boldsymbol{g}}(z)$. All these generators lie on the unit circle at
  the $(M-1)$-th roots of $-1$.
\end{IEEEproof}

Corollary~\ref{crl:corollary1} states that, in thinned ULAs, the number
of columns $L$ of the array steering matrix
${\boldsymbol A}_{\mathcal{I}}(z_1, \ldots, z_L)$, i.e., the number
of identifiable sources, must be strictly less than the number of
physical sensors $M_{\mathcal{I}}$, that is, $L < M_{\mathcal{I}}$~\cite{waxUniqueLocalizationMultiple1989}.
This contrasts with the fully populated ULA, where the Vandermonde
determinant theorem guarantees full column rank of
${\boldsymbol A}_M(z_1, \ldots, z_L)$ for any $L \leq M$ and any set of
distinct generators~\cite{890366Sidiropoulos,
  Heineman1929GeneralizedVandermonde, buck_generalized_1992,
SchlickeweiViola2000GeneralizedVandermonde}. Corollary~\ref{crl:corollary1} can therefore be viewed as an algebraic
counterpart of the spatial sampling theorem, which guarantees that every set of
distinct generators yields a full-column-rank steering matrix whenever the
inter-element spacing satisfies $d \leq \lambda/2$, where $\lambda$ is the
wavelength of the impinging narrowband signal. The result of
Corollary~\ref{crl:corollary1} is consistent with earlier reports
in~\cite{manikasModelingEstimationAmbiguities1998,
abramovichIdentifiabilityManifoldAmbiguity1999}, where ambiguity
generator sets formed by roots of unity were derived for this case
through different arguments.
\begin{corollary}[Maximum-consecutive-lag arrays]\label{crl:corollary2}
  If the number of sensors in the thinned array exceeds the number of
  sources, i.e., $M_{\mathcal{I}} > L$, then there always exists a sparse
  sensor configuration, obtained by removing up to
  $M - M_{\mathcal{I}}$ sensors from the underlying $M$-element ULA, whose
  array manifold is unambiguous, i.e.,
  ${\boldsymbol A}_{\mathcal{I}}(z_1,\ldots,z_L) \in
  \mathbb{C}^{M_{\mathcal{I}} \times L}$ has full column rank for every set
  of distinct generators $\{z_1,\ldots,z_L\}$.
\end{corollary}
\begin{IEEEproof}
  Consider the boundary case $M_{\mathcal{I}} = L + 1$ and choose
  $\bar{\mathcal{I}} = \{2, 3, \ldots, M - L\}$, which corresponds to
  removing $M - L - 1$ consecutive sensors from positions
  $2, 3, \ldots, M - L$ of the underlying ULA. The retained index set is
  \begin{equation}\label{eq:index_set_example}
    \mathcal{I}  \triangleq \{1, M - L + 1, M - L + 2, \ldots, M\}.
  \end{equation}
  The corresponding complementary submatrix
  ${\mathcal{T}}_{\bar{\mathcal{I}}}({\boldsymbol g}) \in
  \mathbb{C}^{(M - L) \times (M - L - 1)}$ contains an upper-triangular $(M-L-1)\times(M-L-1)$ submatrix with $g_0$ on its main diagonal, and therefore has full column rank for any nonzero
  $g_0$. By Theorem~\ref{thm:rank_condition}\,,
  ${\boldsymbol A}_{\mathcal{I}}(z_1,\ldots,z_L) \in
  \mathbb{C}^{(L+1) \times L}$ has full column rank for every set of
  distinct generators $\{z_1,\ldots,z_L\}$.
\end{IEEEproof}
The array defined by the index set $\mathcal{I}$
in~\eqref{eq:index_set_example} is commonly referred to as a
maximum-consecutive-lag array, since it contains a contiguous subarray of
$L$ consecutive sensors~\cite{10835190gini}. This subarray itself forms an
$L$-element ULA and, by the Vandermonde determinant
theorem~\cite{https://doi.org/10.1002/zamm.19870670330HornJohnson}, admits
an unambiguous manifold for $L$ sources. While the result of
Corollary~\ref{crl:corollary2} is therefore well known, it confirms that
Theorem~\ref{thm:rank_condition} reproduces this elementary fact through
its algebraic characterization. The Examples~\ref{exmpl:7b}\,,
\ref{exmpl:8b}\,, and \ref{exmpl:9b} in Section \ref{sec:Examples} will further show that thinned
ULAs with maximum consecutive lag strictly less than $L$, i.e., genuinely
sparse arrays that contain no $L$-element contiguous subarray, can also
be designed to be unambiguous, provided the total number of sensors
$M_{\mathcal{I}}$ exceeds $L$. Such truly sparse configurations
may be preferable when mutual coupling is a concern.
\section{Design of Unambiguous Thinned Arrays}\label{sec:ArrayDesign}
In this section we provide practical guidelines on how the results of
Theorem~\ref{thm:rank_condition} can be used to design unambiguous sparse
arrays.

\begin{corollary}[Forbidden center elements in the removal set]\label{crl:corollary3}
  Assume $2L \geq M + 1$ and define the center index set
  $\mathcal{Q} = \{M - L + 1, \ldots, L\}$. Then, if any element of
  $\mathcal{Q}$ is included in the removal set, i.e., if
  $\mathcal{Q}\, \cap \, \bar{\mathcal{I}} \neq \emptyset$, the generalized
  Vandermonde matrix ${\boldsymbol A}_{\mathcal{I}}(z_1, \ldots, z_L)$ is
  ambiguous: there exists a generator set on the unit circle for which it
  becomes rank deficient.
\end{corollary}
\begin{IEEEproof}
  Choose the coefficient vector $\grave{\boldsymbol g} = [1, 0, \ldots, 0, 1]^{\mathsf{T}}$,
  i.e., $\grave{g}_0 = \grave{g}_L = 1$ and $\grave{g}_\ell = 0$ for
  $\ell = 1, \ldots, L - 1$. From the definition of ${\mathcal{T}}({\boldsymbol g})$
  in~\eqref{def:ToeplitzEven} and~\eqref{def:ToeplitzOdd}, the columns of
  ${\mathcal{T}}(\grave{\boldsymbol g})$ indexed by $\mathcal{Q}$ are
  identically zero. Hence, if $\mathcal{Q}\, \cap\, \bar{\mathcal{I}} \neq \emptyset$,
  the complementary submatrix
  ${\mathcal{T}}_{\bar{\mathcal{I}}}(\grave{\boldsymbol g})$ contains a zero
  column and is therefore rank deficient. By
  \begin{equation}\label{eq:rootsofunity_caseA}
    P_{\grave{\boldsymbol{g}}} (z) = 1 + z^L =  \prod_{\ell = 1}^{L}\!\big(z - e^{{\mathsf j}\pi(2\ell - 1)/L}\big),
  \end{equation}
  cf.~\eqref{eq:conv_b_h_g},
  the roots $\grave{z}_1, \ldots, \grave{z}_L$ of
  $P_{\grave{\boldsymbol g}}(z)$ lie on the unit circle, and by
  Theorem~\ref{thm:rank_condition}\,,
  ${\boldsymbol A}_{\mathcal{I}}(\grave{z}_1, \ldots, \grave{z}_L)$ is rank
  deficient.
\end{IEEEproof}
\begin{corollary}[Mutually exclusive indices, Type A]\label{crl:corollary4a}
  Assume $2L > M - 1$ and define $\mathcal{Q} = \{M - L,\, L + 1\}$. If
  both indices in $\mathcal{Q}$ are contained in the removal set, i.e.,
  $\mathcal{Q} \subseteq \bar{\mathcal{I}}$, then there exists an ambiguous
  generator set $\{\acute{z}_1, \ldots, \acute{z}_L\}$ on the unit circle for
  which the generalized Vandermonde matrix
  ${\boldsymbol A}_{\mathcal{I}}(\acute{z}_1, \ldots, \acute{z}_L)$ is rank
  deficient.
\end{corollary}
\begin{IEEEproof}
  Choose $\acute{\boldsymbol g} = {\boldsymbol 1}$. From~\eqref{def:ToeplitzEven}
  and~\eqref{def:ToeplitzOdd}, columns $M - L$ and $L + 1$ of
  ${\mathcal{T}}(\acute{\boldsymbol g})$ are identical. Hence, if both
  indices are contained in $\bar{\mathcal{I}}$, the complementary submatrix
  ${\mathcal{T}}_{\bar{\mathcal{I}}}(\acute{\boldsymbol g})$ has two
  identical columns and is rank deficient. With the factorization
  \begin{equation}\label{eq:rootsofunity_caseB}
    P_{\acute{\boldsymbol g}}(z) = \sum_{\ell = 0}^{L} z^\ell = \frac{z^{L+1}-1}{z-1}
    =
    \prod_{k=1}^{L}
    \left(z-e^{{\mathsf j}2\pi k/(L+1)}\right),
  \end{equation}
  the roots of
  $P_{\acute{\boldsymbol g}}(z)$ given by $\acute{z}_\ell = e^{{\mathsf j}2\pi \ell/(L+1)}$ ($\ell = 1, \ldots, L$) lie on the unit
  circle, and Theorem~\ref{thm:rank_condition} implies that
  ${\boldsymbol A}_{\mathcal{I}}(\acute{z}_1, \ldots, \acute{z}_L)$ is rank
  deficient.
\end{IEEEproof}
\begin{corollary}[Mutually exclusive indices, Type B]\label{crl:corollary4b}
  Let $p, q$ be integers satisfying, with $N = M - L$,
  \begin{equation}\label{eq:corollary4b-conditions}
    0\leq p < \tfrac{L}{2},
    \quad
    p + 1 \leq q \leq L - p,
    \quad
    N \leq L + q - 2p - 1,
  \end{equation}
  and define the index pair $\mathcal{S}_{p,q} \triangleq \{q,\, L - p + q\}$.
  If $\mathcal{S}_{p,q} \subseteq \bar{\mathcal{I}}$, then there exists an
  ambiguity set $\{\check{z}_1, \ldots, \check{z}_L\}$ of $L$ distinct
  generators on the unit circle for which
  ${\boldsymbol A}_{\mathcal{I}}(\check{z}_1, \ldots, \check{z}_L)$ is rank
  deficient.
\end{corollary}
\begin{IEEEproof}
  Let $\alpha \in \mathbb{C}$ with $|\alpha| = 1$ and define the coefficient
  vector $\check{\boldsymbol g}$ by
  \begin{equation}\label{eq:gcheck_alpha}
    \check{g}_0 = \check{g}_{L-p} = \alpha,
    \qquad
    \check{g}_p = \check{g}_L = \alpha^*,
  \end{equation}
  and $\check{g}_\ell = 0$ otherwise, where for $p = 0$ the two pairs
  in~\eqref{eq:gcheck_alpha} collapse and $\alpha = 1$ is chosen. Since
  $p < L/2 < L - p$, the four indices in~\eqref{eq:gcheck_alpha} are distinct
  for $p \geq 1$. The vector $\check{\boldsymbol g}$ is centro-Hermitian,
  because the index pairs $(0, L)$ and $(p, L - p)$ satisfy
  $\check{g}_\ell = \check{g}_{L - \ell}^*$, and for even $L = 2K$ the middle
  coefficient is $\check{g}_K = 0$, and hence real.

  For $p = 0$ the polynomial $P_{\check{\boldsymbol g}}(z)$ is as
  in~\eqref{eq:rootsofunity_caseA}, with its $L$ roots at the $L$-th roots of
  $-1$, which are distinct and lie on the unit circle. For $p \geq 1$,
  \begin{equation}\label{eq:rootsofunity_caseC}
    P_{\check{\boldsymbol g}}(z)
    = \alpha + \alpha^* z^p + \alpha\, z^{L-p} + \alpha^* z^L
    = \big(\alpha + \alpha^* z^{p}\big)\big(1 + z^{L-p}\big),
  \end{equation}
  so that the generators are the $p$ solutions of $z^p = -\alpha^2$ together
  with the $(L-p)$-th roots of $-1$, all of which lie on the unit circle, and
  where $\alpha$ is chosen such that the two root families do not
  intersect.\,\footnote{Such $\alpha$ always exist: an intersection occurs only
    if $\alpha^2 = -z_0^{\,p}$ for one of the $L - p$ roots $z_0$ of
    $z^{L-p} = -1$, which excludes at most $2(L-p)$ points on the unit circle.
    Explicitly, the intersection occurs exactly for
$\alpha^2 = -e^{{\mathsf j}\pi(2k+1)p/(L-p)}$, $k = 0, \ldots, L-p-1$; the
canonical choice $\alpha = 1$ is among the excluded values whenever $p$ and
$L-p$ contain the same power of two.}
  For every such $\alpha$ the $L$ generators are pairwise distinct, as required
  by (A2), cf.\ Remark~\ref{rmk:repeated}\,.

  We now show that, with this choice of $\check{\boldsymbol g}$, columns
  $q$ and $L - p + q$ of ${\mathcal{T}}(\check{\boldsymbol g})$ are
  identical. The nonzero entries in the first row of
  ${\mathcal{T}}(\check{\boldsymbol g})$ occur at column positions
  $1, p+1, L-p+1, L+1$, with values $\alpha, \alpha^*, \alpha, \alpha^*$,
  respectively. Since ${\mathcal{T}}(\check{\boldsymbol g})$ is Toeplitz, an
  entry appearing in the first row at column position $s + 1$ appears in row
  $r$ at column position $r + s$. Hence column $j$ contains a nonzero entry in
  row $r$ if and only if $r = j - s$ for some shift
  $s \in \{0, p, L - p, L\}$ (for $p = 0$ this reduces to
  $s \in \{0, L\}$).

  Restricting to valid row indices $\{1, \ldots, N\}$, where $N = M - L$,
  the set of nonzero row indices in column $q$ is
  \begin{equation}\label{eq:Set_q}
    \mathcal{S}_q
    = \{q,\, q - p,\, q - L + p,\, q - L\} \cap \{1, \ldots, N\}.
  \end{equation}
  The condition $p + 1 \leq q \leq L - p$ implies
  $1 \leq q - p \leq q$ and $q - L + p \leq 0$, $q - L \leq 0$, so the last two
  candidates are outside the valid range. Hence
  $\mathcal{S}_q = \{q,\, q - p\} \cap \{1, \ldots, N\}$ (for $p = 0$ this
  reduces to $\mathcal{S}_q = \{q\} \cap \{1, \ldots, N\}$).

  For column $L - p + q$, the analogous calculation yields
  \begin{equation}
    \mathcal{S}_{L-p+q}
    = \{L - p + q,\, L - 2p + q,\, q,\, q - p\} \cap \{1, \ldots, N\}.
  \end{equation}
  The condition $N \leq L + q - 2p - 1$ is equivalent to
  $L - 2p + q \geq N + 1$, so $L - 2p + q > N$. Since $p \geq 0$,
  $L - p + q = (L - 2p + q) + p > N$ as well. Therefore the first two
  candidates are outside the valid range, and
  $\mathcal{S}_{L-p+q} = \{q,\, q - p\} \cap \{1, \ldots, N\}
  = \mathcal{S}_q$ (for $p = 0$ this reduces to
  $\mathcal{S}_{L+q} = \{q\} \cap \{1, \ldots, N\} = \mathcal{S}_q$).

  Columns $q$ and $L - p + q$ thus have nonzero entries in exactly the same
  rows, and in these rows the entries agree as well: whenever they lie in
  $\{1, \ldots, N\}$, row $q$ carries $\check{g}_0 = \alpha$ and
  $\check{g}_{L-p} = \alpha$, and row $q - p$ carries $\check{g}_p = \alpha^*$
  and $\check{g}_L = \alpha^*$, respectively (for $p = 0$ only row $q$ occurs, with the common value
  $\check{g}_0 = \check{g}_L = 1$). Hence columns $q$ and $L - p + q$ of
  ${\mathcal{T}}(\check{\boldsymbol g})$ are identical.
  Consequently, if $\mathcal{S}_{p,q} \subseteq \bar{\mathcal{I}}$, the
  complementary submatrix
  ${\mathcal{T}}_{\bar{\mathcal{I}}}(\check{\boldsymbol g})$ contains two
  identical columns and is rank deficient. By
  Theorem~\ref{thm:rank_condition}\,, the corresponding generalized
  Vandermonde matrix
  ${\boldsymbol A}_{\mathcal{I}}(\check{z}_1, \ldots, \check{z}_L)$ is rank
  deficient, where $\check{z}_1, \ldots, \check{z}_L$ are the roots of
  $P_{\check{\boldsymbol g}}(z)$.
\end{IEEEproof}
Corollaries~\ref{crl:corollary3}\,, \ref{crl:corollary4a}\,,
and~\ref{crl:corollary4b} establish that certain index pairs are
\emph{mutually exclusive} for inclusion in the removal set
$\bar{\mathcal{I}}$: if both indices of $\mathcal{Q}$ or of
$\mathcal{S}_{p,q}$ are removed, the resulting array manifold becomes
ambiguous, and the corresponding steering matrix is rank deficient.

For Corollaries~\ref{crl:corollary3} and~\ref{crl:corollary4a}\,, and for
Corollary~\ref{crl:corollary4b} with $p = 0$, the resulting ambiguity generator
sets, i.e., the roots of the polynomials~\eqref{eq:rootsofunity_caseA}
and~\eqref{eq:rootsofunity_caseB}, consist of roots of unity (or of $-1$). Such
roots-of-unity ambiguity patterns have also been reported
in~\cite{manikasModelingEstimationAmbiguities1998}
and~\cite{matterAmbiguitiesDoAEstimation2022}. For $p \geq 1$, in contrast, Corollary~\ref{crl:corollary4b} yields a continuum
of ambiguity sets, parameterized by $\alpha$
in~\eqref{eq:rootsofunity_caseC}: a rotated copy of the $p$-th roots of $-1$
superimposed on the fixed $(L-p)$-th roots of $-1$.

Furthermore, Corollaries~\ref{crl:corollary3}\,, \ref{crl:corollary4a}\,, and~\ref{crl:corollary4b}  can provide upper bounds for the maximum number of sources $L$ for a full-rank thinned ULA steering matrix for a given aperture $M$ and a given number of removed sensors $M_{\bar{\mathcal I}} = M - M_{\mathcal I}$.
\subsection{Fixing predefined generators}
If a subset of generators $\{\acute{z}_1, \ldots, \acute{z}_{L'}\}$ of the polynomial
$P_{\boldsymbol g}(z)$ in~\eqref{eq:Polynomial} is prescribed, the
polynomial can be factorized as
\begin{equation}
  P_{\boldsymbol g}(z) = P_{\boldsymbol f}(z)\, P_{\boldsymbol h}(z),
\end{equation}
where
$ P_{\boldsymbol h}(z)
   \triangleq \sum_{\ell = 0}^{L'} h_{\ell}\, z^{\ell}
  = \prod_{\ell = 1}^{L'} (z - \acute{z}_\ell)$, 
and the remaining $L - L'$ roots of $P_{\boldsymbol g}(z)$ are the roots
of $P_{\boldsymbol f}(z)
 \triangleq \sum_{n = 0}^{L - L'} f_{n}\, z^{n}$. Under this factorization, the Toeplitz matrix
in~\eqref{def:ToeplitzEven}--\eqref{def:ToeplitzOdd} admits the
factorization
\begin{equation}\label{eq:Toeplitz-factorization}
  {\mathcal{T}}({\boldsymbol g}) = {\mathcal{T}}({\boldsymbol h})\,
  {\mathcal{T}}({\boldsymbol f}),
\end{equation}
where
${\mathcal{T}}({\boldsymbol f}) \in \mathbb{C}^{(M - L + L') \times M}$
and
${\mathcal{T}}({\boldsymbol h}) \in \mathbb{C}^{(M - L) \times (M - L + L')}$.
Substituting~\eqref{eq:Toeplitz-factorization} into the definition of the
complementary submatrix~\eqref{def:DefG_0} yields
\begin{equation}\label{eq:ProductG_HF}
  {\mathcal{T}}_{\bar{\mathcal{I}}}({\boldsymbol g})
  = {\mathcal{T}}({\boldsymbol h})\,
  {\mathcal{T}}_{\bar{\mathcal{I}}}({\boldsymbol f}).
\end{equation}

Setting $\bar{L}'  \triangleq L - L'$, the matrix ${\mathcal{T}}({\boldsymbol h})$
admits a convenient null space characterization via the
property~\eqref{eq:subspace_property} applied to the polynomial
$P_{\boldsymbol h}(z)$: the columns of the Vandermonde matrix
${\boldsymbol A}_{M - \bar{L}'}(z_1, \ldots, z_{L'}) \in
\mathbb{C}^{(M - \bar{L}') \times L'}$ form a basis of
$\mathcal{N}({\mathcal{T}}({\boldsymbol h}))$, i.e.,
\begin{equation}\label{eq:TAM}
  {\mathcal{T}}({\boldsymbol h})\,
  {\boldsymbol A}_{M - \bar{L}'}(z_1, \ldots, z_{L'}) = {\boldsymbol 0}.
\end{equation}
This observation leads to the following corollary.

\begin{corollary}[Vandermonde reduction]\label{crl:corollary5}
  Let $L' < L$ and $\bar{L}'  \triangleq  L - L'$, and choose arbitrary distinct points
  $\acute{z}_1, \ldots, \acute{z}_{L'}$ on the unit circle together with a
  cofactor coefficient vector
  ${\boldsymbol f} = [f_0, \ldots, f_{\bar{L}'}]^{\mathsf T}$ with $f_0 \neq 0$. Define
  \begin{equation}\label{eq:Pg-constructive}
    P_{\boldsymbol g}(z)  \triangleq P_{\boldsymbol f}(z)\, P_{\boldsymbol h}(z),
    \qquad
    P_{\boldsymbol h}(z)  \triangleq \prod_{\ell = 1}^{L'} (z - \acute{z}_\ell),
  \end{equation}
  and let $z_1, \ldots, z_L$ denote the $L$ roots of $P_{\boldsymbol g}(z)$,
  assumed distinct. Then the generalized Vandermonde matrix
  ${\boldsymbol A}_{\mathcal{I}}(z_1, \ldots, z_L)$ is rank deficient if and
  only if the $(M - \bar{L}') \times (M_{\bar{\mathcal{I}}} + L')$ augmented
  matrix
  \begin{equation}\label{eq:generalizeVandermonde}
    \mathcal{A}({\boldsymbol f}, \acute{z}_1, \ldots, \acute{z}_{L'})
    \triangleq
    \begin{bmatrix}
      {\mathcal{T}}_{\bar{\mathcal{I}}}({\boldsymbol f}) &
      {\boldsymbol A}_{M - \bar{L}'}(\acute{z}_1, \ldots, \acute{z}_{L'})
    \end{bmatrix}
  \end{equation}
  is rank deficient.
\end{corollary}
\begin{IEEEproof}
  By Theorem~\ref{thm:rank_condition}\,, rank deficiency of
  ${\boldsymbol A}_{\mathcal{I}}(z_1, \ldots, z_L)$ is equivalent to the
  existence of a nonzero ${\boldsymbol y}_{\bar{\mathcal{I}}}$ with
  $\mathcal{T}_{\bar{\mathcal{I}}}(\boldsymbol{g})\,{\boldsymbol y}_{\bar{\mathcal{I}}}
  = \boldsymbol{0}$. Substituting~\eqref{eq:ProductG_HF}, this means
  $\mathcal{T}_{\bar{\mathcal{I}}}(\boldsymbol{f})\,{\boldsymbol y}_{\bar{\mathcal{I}}}
  \in \mathcal{N}(\mathcal{T}(\boldsymbol{h}))$, which by~\eqref{eq:TAM}
  is equivalent to the existence of a $\boldsymbol{\beta}$ with
  \begin{equation}\label{eq:cor8-proof}
    \mathcal{T}_{\bar{\mathcal{I}}}(\boldsymbol{f})\,{\boldsymbol y}_{\bar{\mathcal{I}}}
    = {\boldsymbol A}_{M - \bar{L}'}(\acute{z}_1, \ldots, \acute{z}_{L'})\,
    \boldsymbol{\beta},
  \end{equation}
  that is, with
  $[{\boldsymbol y}_{\bar{\mathcal{I}}}^{\mathsf{T}},\,
  -\boldsymbol{\beta}^{\mathsf{T}}]^{\mathsf{T}}$
  a null vector of
  $\mathcal{A}(\boldsymbol{f}, \acute{z}_1, \ldots, \acute{z}_{L'})$
  in~\eqref{eq:generalizeVandermonde}.

  If ${\boldsymbol A}_{\mathcal{I}}(z_1, \ldots, z_L)$ is rank deficient, then a nonzero null vector
  ${\boldsymbol y}_{\bar{\mathcal{I}}}$ exists and thus also $\mathcal{A}(\boldsymbol{f}, \acute{z}_1, \ldots,
  \acute{z}_{L'})$ is rank deficient. Conversely, if
  $\mathcal{A}(\boldsymbol{f}, \acute{z}_1, \ldots, \acute{z}_{L'})$ is rank
  deficient, then any nonzero null vector
  $[{\boldsymbol y}_{\bar{\mathcal{I}}}^{\mathsf{T}},\,
  -\boldsymbol{\beta}^{\mathsf{T}}]^{\mathsf{T}}$
  necessarily has ${\boldsymbol y}_{\bar{\mathcal{I}}} \neq \boldsymbol{0}$:
  otherwise~\eqref{eq:cor8-proof} would give
  ${\boldsymbol A}_{M - \bar{L}'}\boldsymbol{\beta} = \boldsymbol{0}$ which, since
  the prescribed generators are distinct and
  ${\boldsymbol A}_{M - \bar{L}'}(\acute{z}_1, \ldots, \acute{z}_{L'})$ full
  column rank, would imply that  $\boldsymbol{\beta} = \boldsymbol{0}$ as well, contradicting
  the assumption that the null vector is nonzero. Reversing the above chain of
  equivalences then yields the rank deficiency of
  ${\boldsymbol A}_{\mathcal{I}}(z_1, \ldots, z_L)$.
\end{IEEEproof}
Corollary~\ref{crl:corollary5} provides a direct relation between the
rank of the generalized Vandermonde matrix
${\boldsymbol A}_{\mathcal{I}}(z_1, \ldots, z_L)$ of the thinned array and
that of the full Vandermonde matrix
${\boldsymbol A}_{M - \bar{L}'}(\acute{z}_1, \ldots, \acute{z}_{L'})$ of a
ULA with $M - \bar{L}'$ elements, for any arbitrary subset of $L'$
generators. This reveals a structural link between sparse and fully
populated arrays: characterizing ambiguities in the thinned array reduces
to characterizing rank deficiencies of an augmented full-ULA steering
matrix. The equivalence is, however, stated for a \emph{fixed} choice of the
$L'$ prescribed generators. Full column rank of
$\mathcal{A}({\boldsymbol f}, \acute{z}_1, \ldots, \acute{z}_{L'})$ for all
admissible ${\boldsymbol f}$ therefore excludes only those ambiguity sets that
\emph{contain} $\acute{z}_1, \ldots, \acute{z}_{L'}$, and statements of
unambiguity for $L' \geq 2$ are to be understood as conditional on that choice.
No generality is lost for $L' = 1$ with $\acute{z}_1 = 1$, since by
Corollary~\ref{crl:rotation_invariance} every ambiguity set can be rotated so
that one of its generators equals one. For $L' \geq 2$ the rotational freedom is already exhausted, so the remaining
prescribed generators cannot be matched to a given ambiguity set by rotation;
completeness is instead recovered by sweeping them over the field of view, as
in Example~\ref{exmpl:3}\,.

\section{Analysis and Design Examples}\label{sec:Examples}
In this section, we illustrate the algebraic framework introduced in
Section~\ref{sec:MainResults} on representative sparse linear array
configurations. The examples cover both the \emph{analysis} task of fully
characterizing the ambiguity sets of a given thinned ULA and the
\emph{design} task of constructing sparse arrays with favorable ambiguity
behavior.

\begin{example}\label{exmpl:1}
  \begin{figure}[!tb]
    \centering
    \begin{tikzpicture}[
        x=0.8cm, y=0.8cm,
        sensor/.style={circle, fill=black, draw=black, inner sep=0pt, minimum size=5pt},
        removed/.style={circle, fill=white, draw=black, inner sep=0pt, minimum size=5pt, line width=0.5pt},
      ]
      \draw[gray!60, thin] (0.5,0) -- (9.5,0);

      \foreach \p in {1,2,3,5,6,7,8,9}{
        \node[sensor] at (\p,0) {};
      }

      \node[removed] at (4,0) {};

      \node[anchor=north, font=\footnotesize] at (1,-0.2) {$1$};
      \node[anchor=north, font=\footnotesize] at (4,-0.2) {$i$};
      \node[anchor=north, font=\footnotesize] at (9,-0.2) {$9$};
    \end{tikzpicture}
      \caption{Thinned ULA with $\mathcal{I} = \{1, 2, \ldots, 9\} \setminus \{i\}$
    obtained by removing a single position $i \in \{2, 3, \ldots, 8\}$ from a
    nine-element ULA (Corollary~\ref{crl:corollary3}\,,
    Examples~\ref{exmpl:1}--\ref{exmpl:3}\,\,). The configuration is
  \emph{ambiguous} for $L = 8$, for every choice of $i$.}
    \label{fig:single_hole_ula}
  \end{figure}

  Consider a thinned ULA with $M_{\mathcal{I}} = 8$ sensors and a single
  hole at position $i$, $i \in \{2, 3, \ldots, 8\}$, so that
  $\mathcal{I} = \{1, 2, \ldots, 9\} \setminus \{i\}$ and $M = 9$. The array
  is illustrated in Fig.~\ref{fig:single_hole_ula} for an arbitrary choice
  of $i$. For $L = 8$ sources, the Toeplitz matrix is the single row
  $\mathcal{T}(\boldsymbol{g}) = [g_0,\, g_1,\, g_2,\, g_3,\, g_4,\, g_3^*,\,
  g_2^*,\, g_1^*,\, g_0^*]$, where $g_0, \ldots, g_3 \in \mathbb{C}$ with
  $g_0 \neq 0$ and $g_4 \in \mathbb{R}$. The complementary submatrix
  in~\eqref{def:DefG_0} reduces to the scalar
  $\mathcal{T}_{\bar{\mathcal{I}}}(\boldsymbol{g}) = g_{i-1}$, which vanishes
  for $g_{i-1} = 0$. Since the centro-Hermitian property
  $g_{i-1} = g_{9-i}^*$ couples the coefficients pairwise, and the polynomial
  $P_{\boldsymbol{g}}(z)$ in~\eqref{eq:Polynomial} can always be chosen to
  have all roots on the unit circle (for instance, via the constructions
  in~\eqref{eq:rootsofunity_caseA}), the
  condition $g_{i-1} = 0$ is satisfiable. By
  Theorem~\ref{thm:rank_condition}\,, ambiguities always exist. This is a
  special case of Corollary~\ref{crl:corollary1}\,.
\end{example}
\begin{example}\label{exmpl:2}
  \begin{figure}[!tb]
    \centering
    \begin{tikzpicture}[
        x=0.8cm, y=0.8cm,
        sensor/.style={circle, fill=black, draw=black, inner sep=0pt, minimum size=5pt},
        removed/.style={circle, fill=white, draw=black, inner sep=0pt, minimum size=5pt, line width=0.5pt},
      ]
      \draw[gray!60, thin] (0.5,0) -- (9.5,0);

      \foreach \p in {1,3,4,5,6,7,8,9}{
        \node[sensor] at (\p,0) {};
      }

      \node[removed] at (2,0) {};

      \foreach \p in {1,...,9}{
        \node[anchor=north, font=\footnotesize] at (\p,-0.2) {$\p$};
      }
    \end{tikzpicture}
    \caption{Thinned ULA with $\mathcal{I} = \{1, 3, 4, 5, 6, 7, 8, 9\}$ obtained
    by removing position $2$ from a nine-element ULA (Example~\ref{exmpl:2}\,\,). It
  is \emph{unambiguous} for $L \leq 7$.}
    \label{fig:hole_at_2}
  \end{figure}

  Consider again a thinned ULA with $M_{\mathcal{I}} = 8$ sensors but with
  the single hole at position $i_1 = 2$, i.e.,
  $\mathcal{I} = \{1, 3, 4, 5, 6, 7, 8, 9\}$ and $L = 7$ sources. The array
  is illustrated in Fig.~\ref{fig:hole_at_2}\,. The Toeplitz matrix becomes
  \begin{equation}\label{eq:example2}
    \mathcal{T}(\boldsymbol{g}) =
    \begin{bmatrix}
      g_0 & g_1 & g_2 & g_3 & g_3^* & g_2^* & g_1^* & g_0^* & 0\\
      0   & g_0 & g_1 & g_2 & g_3   & g_3^* & g_2^* & g_1^* & g_0^*
    \end{bmatrix},
  \end{equation}
  and the complementary submatrix reduces to
  $\mathcal{T}_{\bar{\mathcal{I}}}(\boldsymbol{g}) = [g_1,\, g_0]^{\mathsf{T}}$.
  Since $g_0 \neq 0$ by definition, this column vector has full column rank,
  i.e., it cannot drop rank below one. Hence, by
  Theorem~\ref{thm:rank_condition}\,, the array manifold
  ${\boldsymbol A}_{\mathcal{I}}(z_1, \ldots, z_L)$ is unambiguous for
  $L = 7$, and the same conclusion extends to $L < 7$.

  The array defined by $\mathcal{I}$ is a maximum-consecutive-lag array,
  so Example~\ref{exmpl:2} is a special case of
  Corollary~\ref{crl:corollary2}\,.
\end{example}
\begin{example}\label{exmpl:3}
  \begin{figure}[!tb]
    \centering
    \begin{tikzpicture}[
        x=0.8cm, y=0.8cm,
        sensor/.style={circle, fill=black, draw=black, inner sep=0pt, minimum size=5pt},
        removed/.style={circle, fill=white, draw=black, inner sep=0pt, minimum size=5pt, line width=0.5pt},
      ]
      \draw[gray!60, thin] (0.5,0) -- (9.5,0);

      \foreach \p in {1,2,4,5,6,7,8,9}{
        \node[sensor] at (\p,0) {};
      }

      \node[removed] at (3,0) {};

      \foreach \p in {1,...,9}{
        \node[anchor=north, font=\footnotesize] at (\p,-0.2) {$\p$};
      }
    \end{tikzpicture}
   \caption{Thinned ULA with $\mathcal{I} = \{1, 2, 4, 5, 6, 7, 8, 9\}$ obtained
    by removing position $3$ from a nine-element ULA (Example~\ref{exmpl:3}\,\,). In
    contrast to Fig.~\ref{fig:hole_at_2}\,, it is \emph{ambiguous} for $L = 7$ and
  admits a continuum of ambiguity sets.}
    \label{fig:hole_at_3}
  \end{figure}

  Consider a thinned ULA with $M_{\mathcal{I}} = 8$ sensors and a single
  hole at position $i_1 = 3$, i.e.,
  $\mathcal{I} = \{1, 2, 4, 5, 6, 7, 8, 9\}$ and $L = 7$ sources, as shown
  in Fig.~\ref{fig:hole_at_3}\,. The Toeplitz matrix
  $\mathcal{T}(\boldsymbol{g})$ has the same form as
  in~\eqref{eq:example2}, and the complementary submatrix becomes
  $\mathcal{T}_{\bar{\mathcal{I}}}(\boldsymbol{g}) = [g_2,\, g_1]^{\mathsf{T}}$.
  This vector vanishes if and only if $g_1 = g_2 = 0$, while $g_0$ and $g_3$
  remain free. Hence the ambiguity sets of the array are characterized by
  the unit-circle roots of the polynomial
  \begin{equation}\label{eq:poly}
    P_{\boldsymbol{g}}(z) = g_0 + g_3\, z^3 + g_3^*\, z^4 + g_0^*\, z^7,
  \end{equation}
  in the sense that any choice of $g_0, g_3 \in \mathbb{C}$ (with
  $g_0 \neq 0$) for which all roots of $P_{\boldsymbol{g}}(z)$ lie on the
  unit circle yields an ambiguity set at which the generalized Vandermonde
  matrix ${\boldsymbol A}_{\mathcal{I}}(z_1, \ldots, z_L)$ becomes rank
  deficient. In practice, however, it is not obvious which choices of
  $g_0, g_3$ admit such a root configuration, so we turn to the augmented
  matrix formulation of Corollary~\ref{crl:corollary5}\,.

  \begin{figure}[!tb]
    \centering
    \setlength{\figurewidth}{0.8\linewidth}
    \includegraphics[width=\columnwidth]{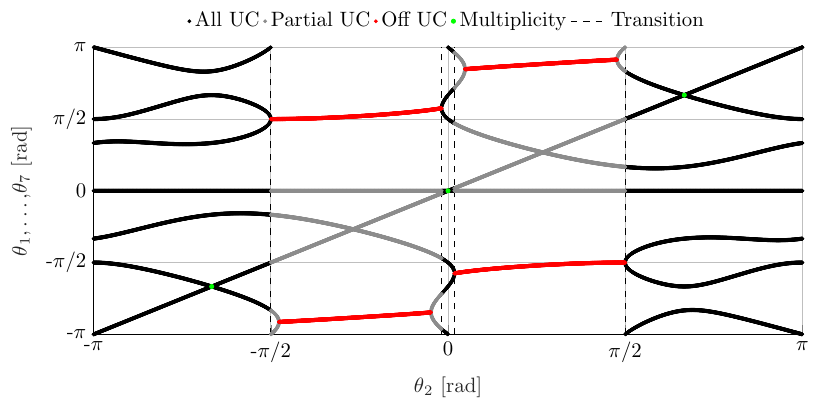}
    \vspace{-1em}
\caption{Ambiguity sets for the thinned ULA in Fig.\ref{fig:hole_at_3} with $L=7$ sources. The root angles $\theta_i=\arg z_i$ of $P_{{\boldsymbol f},\acute{z}_2}(z)$ in\eqref{eq:Poly_f_z2} for $f_0=1$ are shown versus $\theta_2=\arg\acute{z}_2$ across the field of view. Black dots indicate cases where all roots lie on the unit circle; red dots indicate off-circle roots; gray dots indicate unit-circle roots when not all roots lie on the unit circle; and green dots indicate generators with multiplicity.}
    \label{fig:thin_ula_miss3}
  \end{figure}

  Fix two prescribed generators $\acute{z}_1 = 1$ and
  $\acute{z}_2 = e^{{\mathsf j}\theta_2}$, with $\acute{z}_2 \neq 1$. By
  Corollary~\ref{crl:corollary5}\,, the array manifold is ambiguous if and
  only if the $4 \times 3$ augmented matrix
  \begin{equation}\label{eq:Adef_exmpl3}
    \mathcal{A}({\boldsymbol f}, \acute{z}_1, \acute{z}_2)
    =
    \begin{bmatrix}
      f_2 & 1 & 1            \\
      f_1 & 1 & \acute{z}_2     \\
      f_0 & 1 & \acute{z}_2^2   \\
      0   & 1 & \acute{z}_2^3
    \end{bmatrix}
  \end{equation}
  is rank deficient. Since $\acute{z}_2 \neq 1$, the second and third
  columns of $\mathcal{A}({\boldsymbol f}, \acute{z}_1, \acute{z}_2)$ are
  linearly independent. Hence rank deficiency holds if and only if the
  first column lies in the span of the other two, i.e., if there exist
  $\alpha_1, \alpha_2 \in \mathbb{C}$ such that
  \begin{equation}
    \begin{bmatrix}
      f_2\\ f_1\\ f_0\\ 0
    \end{bmatrix}
    =
    \alpha_1\,
    \begin{bmatrix}
      1\\1\\1\\1
    \end{bmatrix}
    +
    \alpha_2\,
    \begin{bmatrix}
      1\\ \acute{z}_2\\ \acute{z}_2^2\\ \acute{z}_2^3
    \end{bmatrix}.
  \end{equation}
  From the last row, $\alpha_1 = -\alpha_2\, \acute{z}_2^3$. Substituting
  into the remaining rows yields
  \begin{equation}
    f_2 = \alpha_2(1 - \acute{z}_2^3),
    \quad
    f_1 = \alpha_2\, \acute{z}_2(1 - \acute{z}_2^2),
    \quad
    f_0 = \alpha_2\, \acute{z}_2^2(1 - \acute{z}_2).
  \end{equation}
  Since $f_0 \neq 0$ and $\acute{z}_2 \neq 1$, we obtain
  \begin{equation}
    \alpha_2 = \frac{f_0}{\acute{z}_2^2\,(1 - \acute{z}_2)},
  \end{equation}
  and the rank-deficient (i.e., ambiguous) case is characterized by
  \begin{equation}\label{eq:f1_f2_exmpl3}
    f_1 = f_0\, \frac{1 + \acute{z}_2}{\acute{z}_2},
    \qquad
    f_2 = f_0\, \frac{1 + \acute{z}_2 + \acute{z}_2^2}{\acute{z}_2^2}.
  \end{equation}
  The associated centro-Hermitian polynomial reads
  \begin{equation}\label{eq:Poly_f_z2} 
P_{{\boldsymbol f}, \acute{z}_2}(z) \triangleq  f_0 + f_1 z + f_2 z^2 + f_2^* z^3 + f_1^* z^4 + f_0^* z^5.
  \end{equation}
The array manifold is ambiguous
  precisely at those $f_0$ and  $\acute{z}_2 = e^{{\mathsf j}\theta_2}$,
  $\acute{z}_2 \neq 1$, for which all roots of
  $P_{{\boldsymbol f}, \acute{z}_2}(z)$ lie on the unit circle.
  Figure~\ref{fig:thin_ula_miss3} displays the corresponding values of
  $\theta_2$ over the field of view for $f_0 = 1$.
\end{example}
\begin{example}\label{exmpl:4a}
  \begin{figure}[!tb]
    \centering
    \begin{tikzpicture}[
        x=0.8cm, y=0.8cm,
        sensor/.style={circle, fill=black, draw=black, inner sep=0pt, minimum size=5pt},
        removed/.style={circle, fill=white, draw=black, inner sep=0pt, minimum size=5pt, line width=0.5pt},
      ]
      \draw[gray!60, thin] (0.5,0) -- (9.5,0);

      \foreach \p in {1,4,5,6,7,8,9}{
        \node[sensor] at (\p,0) {};
      }

      \foreach \p in {2,3}{
        \node[removed] at (\p,0) {};
      }

      \foreach \p in {1,...,9}{
        \node[anchor=north, font=\footnotesize] at (\p,-0.2) {$\p$};
      }
    \end{tikzpicture}
  \caption{Thinned ULA with $\mathcal{I} = \{1, 4, 5, 6, 7, 8, 9\}$ obtained by
    removing $\bar{\mathcal{I}} = \{2, 3\}$ from a nine-element ULA. The
    configuration is \emph{ambiguous} for $L = 7$ (Example~\ref{exmpl:4a}\,) and
  \emph{unambiguous} for $L \leq 6$ (Example~\ref{exmpl:4b}\,).}
    \label{fig:holes_at_2_3}
  \end{figure}

  Consider a thinned ULA with $M_{\mathcal{I}} = 7$ sensors and two adjacent
  holes at positions $i_1 = 2$ and $i_2 = 3$, i.e.,
  $\mathcal{I} = \{1, 4, 5, 6, 7, 8, 9\}$ and $L = 7$ sources, as depicted
  in Fig.~\ref{fig:holes_at_2_3}\,. The Toeplitz matrix
  $\mathcal{T}(\boldsymbol{g})$ has the same form as
  in~\eqref{eq:example2}, and the complementary submatrix
  \eqref{def:DefG_0} becomes
  \begin{equation}
    \mathcal{T}_{\bar{\mathcal{I}}}(\boldsymbol{g})
    =
    \begin{bmatrix}
      g_1 & g_2 \\
      g_0 & g_1
    \end{bmatrix}.
  \end{equation}
  This matrix is rank deficient if and only if
  $\det\big(\mathcal{T}_{\bar{\mathcal{I}}}(\boldsymbol{g})\big)
  = g_1^2 - g_0 g_2 = 0$. Setting $g_0 = 1$ without loss of generality
  (after an appropriate rotation of the roots), the condition reduces to
  $g_2 = g_1^2$, i.e., for $g_1 = |g_1|\, e^{{\mathsf j}\phi_1}$,
  $g_2 = |g_1|^2\, e^{{\mathsf j}2\phi_1}$. The corresponding polynomial
  becomes
  \begin{equation}\label{eq:poly4}
    \begin{aligned}
      P_{\boldsymbol{g}}(z)
      &= 1
      + |g_1|\, e^{{\mathsf j}\phi_1}\, z
      + |g_1|^2\, e^{{\mathsf j}2\phi_1}\, z^2
      + g_3\, z^3
      + g_3^*\, z^4 \\
      &\quad
      + |g_1|^2\, e^{-{\mathsf j}2\phi_1}\, z^5
      + |g_1|\, e^{-{\mathsf j}\phi_1}\, z^6
      + z^7,
    \end{aligned}
  \end{equation}
  where $g_1$ and $g_3 = |g_3|\, e^{{\mathsf j}\phi_3}$ are free parameters.
  Ambiguities exist whenever all roots of~\eqref{eq:poly4} lie on the unit
  circle. However, fully characterizing the parameter range
  $(|g_1|, \phi_1, |g_3|, \phi_3)$ for which this holds, and devising a
  constructive procedure to enumerate the ambiguity sets, appears difficult
  through direct analysis of~\eqref{eq:poly4}.

  We therefore turn to Corollary~\ref{crl:corollary5}\,. Fix three prescribed
  generators on the unit circle: $\acute{z}_1 = 1$,
  $\acute{z}_2 = e^{{\mathsf j}\theta_2}$, and
  $\acute{z}_3 = e^{{\mathsf j}\theta_3}$, with $\acute{z}_i \neq \acute{z}_j$
  for $i \neq j$. The cofactor coefficient vector is
  $\boldsymbol{f} = [f_0, f_1, f_2, f_1^*, f_0^*]^{\mathsf{T}}$ with
  $f_0, f_1 \in \mathbb{C}$ and $f_2 \in \mathbb{R}$ (from the
    centro-Hermitian condition $f_\ell = f_{L - L' - \ell}^*$ with
  $L - L' = 4$). The augmented matrix
  in~\eqref{eq:generalizeVandermonde} is the $5 \times 5$ matrix
  \begin{equation}\label{eq:Adef_exmpl4a}
    \mathcal{A}(\boldsymbol{f}, \acute{z}_1,\acute{z}_2, \acute{z}_3)
    =
    \begin{bmatrix}
      f_1 & f_2 & 1 & 1               & 1               \\
      f_0 & f_1 & 1 & \acute{z}_2      & \acute{z}_3      \\
      0   & f_0 & 1 & \acute{z}_2^{2}  & \acute{z}_3^{2}  \\
      0   & 0   & 1 & \acute{z}_2^{3}  & \acute{z}_3^{3}  \\
      0   & 0   & 1 & \acute{z}_2^{4}  & \acute{z}_3^{4}
    \end{bmatrix}.
  \end{equation}
  By Corollary~\ref{crl:corollary5}\,, the array manifold is ambiguous if and
  only if $\mathcal{A}(\boldsymbol{f}, \acute{z}_2, \acute{z}_3)$ is rank
  deficient.

  Gaussian elimination on~\eqref{eq:Adef_exmpl4a} yields a rank-deficiency
  condition that can be expressed in terms of the complete homogeneous
  symmetric polynomials in the prescribed generators
  $\acute{z}_1 = 1, \acute{z}_2, \acute{z}_3$:
  \begin{align}
    S_1 & \triangleq 1 + \acute{z}_2 + \acute{z}_3,\\
    S_2 & \triangleq 1 + \acute{z}_2 + \acute{z}_3 + \acute{z}_2^2
    + \acute{z}_2 \acute{z}_3 + \acute{z}_3^2.
  \end{align}
  Specifically, $\mathcal{A}(\boldsymbol{f},\acute{z}_1, \acute{z}_2, \acute{z}_3)$ is
  rank deficient if and only if
  \begin{equation}\label{eq:simple_procedure}
    f_2 = \frac{f_1^2}{f_0} - f_1\, S_1^* + f_0\, S_2^*.
  \end{equation}

  A practical construction of rank-deficient matrices
  $\mathcal{A}(\boldsymbol{f}, \acute{z}_1, \acute{z}_2, \acute{z}_3)$ proceeds as
  follows. First, fix two distinct unit-modulus generators
  $\acute{z}_2, \acute{z}_3 \neq 1$. Setting $\gamma  \triangleq  f_1 / f_0$,
  condition~\eqref{eq:simple_procedure} can be rewritten as
  \begin{equation}
    \frac{f_2}{f_0} = \gamma^2 - \gamma\, S_1^* + S_2^*,
  \end{equation}
  which is a quadratic function in $\gamma$. Choosing $f_0$ on the unit circle and
  $f_2 \in [-6, 6]$ real (to satisfy the centro-Hermitian constraint), the
  quadratic function yields up to two candidate values for $\gamma$, and hence up
  to two candidate values $f_1 = \gamma\, f_0$, each producing a rank-4
  matrix $\mathcal{A}(\boldsymbol{f}, \acute{z}_1, \acute{z}_2, \acute{z}_3)$.

  To select solutions that yield the additional generators on the unit circle,
  one subsequently tests the centro-Hermitian polynomial
  \begin{equation}
    F(z)  \triangleq f_0 + f_1 z + f_2 z^2 + f_1^* z^3 + f_0^* z^4.
  \end{equation}
  If all four roots of $F(z)$ lie on the unit circle, the corresponding
  parameters $(\boldsymbol{f}, \acute{z}_2, \acute{z}_3)$ are accepted;
  otherwise, $f_0$ and/or $f_2$ are adjusted and the procedure is repeated.

  Numerical experiments show that, for any choice of three distinct
  prescribed generators $\acute{z}_1 = 1, \acute{z}_2, \acute{z}_3$ on the
  unit circle, parameters $(f_0, f_1, f_2)$ satisfying the above
  construction can be found such that $F(z)$ has four distinct unit-circle
  roots $\acute{z}_4, \acute{z}_5, \acute{z}_6, \acute{z}_7$, which together
  with $\acute{z}_1, \acute{z}_2, \acute{z}_3$ form an ambiguity set of
  seven distinct generators on the unit circle. By
  Corollary~\ref{crl:corollary5}\,, the generalized Vandermonde matrix
  ${\boldsymbol A}_{\mathcal{I}}(\acute{z}_1, \ldots, \acute{z}_7)$ is rank
  deficient.  This is consistent with Corollary~\ref{crl:corollary1}\,, which states that
  ambiguities always exist whenever $L \geq M_{\mathcal{I}}$, and the present
  example is the boundary case $L = M_{\mathcal{I}} = 7$.
\end{example}
\begin{example}\label{exmpl:4b}
  Consider again the thinned ULA of Fig.~\ref{fig:holes_at_2_3} with
  $M_{\mathcal{I}} = 7$ sensors and removal set
  $\bar{\mathcal{I}} = \{2, 3\}$, but now for $L = 6$ sources. The Toeplitz
  matrix becomes
  \begin{equation}\label{eq:example5}
    \mathcal{T}(\boldsymbol{g})
    =
    \begin{bmatrix}
      g_0 & g_1 & g_2 & g_3 & g_2^* & g_1^* & g_0^* & 0     & 0\\
      0   & g_0 & g_1 & g_2 & g_3   & g_2^* & g_1^* & g_0^* & 0\\
      0   & 0   & g_0 & g_1 & g_2   & g_3   & g_2^* & g_1^* & g_0^*
    \end{bmatrix},
  \end{equation}
  and the complementary submatrix is
  \begin{equation}
    \mathcal{T}_{\bar{\mathcal{I}}}(\boldsymbol{g})
    =
    \begin{bmatrix}
      g_1 & g_2 \\
      g_0 & g_1 \\
      0   & g_0
    \end{bmatrix}.
  \end{equation}
  The $2 \times 2$ submatrix formed by the last two rows is upper
  triangular with determinant $g_0^2 \neq 0$, so
  $\mathcal{T}_{\bar{\mathcal{I}}}(\boldsymbol{g})$ has full column rank
  unconditionally. By Theorem~\ref{thm:rank_condition}\,, the array manifold
  ${\boldsymbol A}_{\mathcal{I}}(z_1, \ldots, z_L)$ is unambiguous for
  $L = 6$, and the same conclusion extends to $L < 6$. This example is an
  instance of Corollary~\ref{crl:corollary2}\,: removing the two consecutive
  sensors at positions $\{2, 3\}$ leaves the retained set
  $\mathcal{I} = \{1\} \cup \{4, 5, \ldots, 9\}$, which contains the
  $6$-element contiguous ULA subarray $\{4, 5, \ldots, 9\}$, sufficient to
  guarantee unambiguous identification of $L \leq 6$ sources.
\end{example}
\begin{example}\label{exmpl:6}
  \begin{figure}[!tb]
    \centering
    \begin{tikzpicture}[
        x=0.8cm, y=0.8cm,
        sensor/.style={circle, fill=black, draw=black, inner sep=0pt, minimum size=5pt},
        removed/.style={circle, fill=white, draw=black, inner sep=0pt, minimum size=5pt, line width=0.5pt},
      ]
      \draw[gray!60, thin] (0.5,0) -- (9.5,0);

      \foreach \p in {1,3,5,6,7,8,9}{
        \node[sensor] at (\p,0) {};
      }

      \foreach \p in {2,4}{
        \node[removed] at (\p,0) {};
      }

      \foreach \p in {1,...,9}{
        \node[anchor=north, font=\footnotesize] at (\p,-0.2) {$\p$};
      }
    \end{tikzpicture}
    \caption{Thinned ULA considered in Example~\ref{exmpl:6} for $L = 6$
      sources with $M_{\mathcal{I}} = 7$ sensors and two holes at positions
      $i_1 = 2$ and $i_2 = 4$, i.e., $\mathcal{I} = \{1, 3, 5, 6, 7, 8, 9\}$.  The
      configuration is \emph{ambiguous} for $L = 6$, with a continuous two-parameter
    family of ambiguity sets.}
    \label{fig:holes_at_2_4}
  \end{figure}

  Consider the thinned ULA depicted in Fig.~\ref{fig:holes_at_2_4}\,, with
  $M_{\mathcal{I}} = 7$ sensors and two holes at positions $i_1 = 2$ and
  $i_2 = 4$, i.e., $\mathcal{I} = \{1, 3, 5, 6, 7, 8, 9\}$ and $L = 6$
  sources. The complementary Toeplitz submatrix has the form
  \begin{equation}\label{eq:Exmpl6T_I}
    \mathcal{T}_{\bar{\mathcal{I}}}(\boldsymbol{g})
    =
    \begin{bmatrix}
      g_1 & g_3 \\
      g_0 & g_2 \\
      0   & g_1
    \end{bmatrix}.
  \end{equation}
  Setting $g_0 = 1$ without loss of generality (by an appropriate rotation
  of the roots), we analyze the rank of~\eqref{eq:Exmpl6T_I}. Suppose that
  $\mathcal{T}_{\bar{\mathcal{I}}}(\boldsymbol{g})$ is column-rank
  deficient. Then there exists a nonzero vector
  $\boldsymbol{\alpha}=[\alpha_1,\alpha_2]^{\mathsf T}$ such that
  \begin{equation}
    \begin{bmatrix}
      g_1 & g_3 \\
      1   & g_2 \\
      0   & g_1
    \end{bmatrix}
    \begin{bmatrix}
      \alpha_1\\\alpha_2
    \end{bmatrix}
    =\boldsymbol{0}.
  \end{equation}
  If $\alpha_2=0$, then the second row gives $\alpha_1=0$, contradicting
  $\boldsymbol{\alpha}\neq \boldsymbol{0}$. Hence $\alpha_2\neq 0$. The
  third row therefore implies $g_1=0$, and the first row then gives
  $g_3\alpha_2=0$, so that $g_3=0$. Thus rank deficiency implies
  $g_1=g_3=0$.

  Substituting $g_1=g_3=0$ back in \eqref{eq:Exmpl6T_I}, then
  $\mathcal{T}_{\bar{\mathcal{I}}}(\boldsymbol{g})$ has only one nonzero
  row, namely $[1,g_2]$, and hence has rank one for any
  $g_2\in\mathbb{C}$.  Therefore, the rank-drop condition is
  \begin{equation}\label{eq:rankcond_exmpl6}
    g_1 = g_3 = 0, \quad g_2 \in \mathbb{C} \text{ free}.
  \end{equation}

  Under~\eqref{eq:rankcond_exmpl6} the coefficient vector takes the form
  $\boldsymbol{g} = [1, 0, g_2, 0, g_2^*, 0, 1]^{\mathsf{T}}$, and the
  polynomial $P_{\boldsymbol{g}}(z)$ becomes a polynomial in $w = z^2$,
  namely
  \begin{equation}
    P_{\boldsymbol{g}}(z) = 1 + g_2 z^2 + g_2^* z^4 + z^6
    = 1 + g_2 w + g_2^* w^2 + w^3.
  \end{equation}
  For any choice of two angles
  $\theta_1, \theta_2 \in [-\pi, \pi)$ with
  $\theta_1 \neq \theta_2 \pmod{\pi}$, setting
  \begin{equation}\label{eq:g2_choice}
    g_2  \triangleq -\big(e^{-{\mathsf j}2\theta_1} + e^{-{\mathsf j}2\theta_2}
    - e^{{\mathsf j}2(\theta_1 + \theta_2)}\big)
  \end{equation}
  yields the factorization
  \begin{equation}
    P_{\boldsymbol{g}}(z)
    = (z^2 - e^{{\mathsf j}2\theta_1})\,(z^2 - e^{{\mathsf j}2\theta_2})\,
    (z^2 + e^{-{\mathsf j}2(\theta_1 + \theta_2)}),
  \end{equation}
  so that all six roots lie on the unit circle and, generically, are
  distinct. The roots are listed in Table~\ref{tab:roots6}\,. By
  Theorem~\ref{thm:rank_condition}\,, the corresponding generalized
  Vandermonde matrix
  ${\boldsymbol A}_{\mathcal{I}}(z_1, \ldots, z_6)$ is rank deficient, and
  the array manifold therefore admits a continuous two-parameter family of
  ambiguity sets indexed by $(\theta_1, \theta_2)$.

  \begin{table}[!tb]
    \centering
    \caption{Roots of
      $P_{\boldsymbol{g}}(z) = 1 + g_2 z^2 + g_2^* z^4 + z^6$ in
    Example~\ref{exmpl:6}\,, with $g_2$ given by~\eqref{eq:g2_choice}.}
    \label{tab:roots6}
    \renewcommand{\arraystretch}{1.4}
    \setlength{\tabcolsep}{4pt}
    \begin{tabular}{@{}cccccc@{}}
      \toprule
      $z_1$ & $z_2$ & $z_3$ & $z_4$ & $z_5$ & $z_6$ \\
      \midrule
      $e^{{\mathsf j}\theta_1}$
      & $e^{{\mathsf j}\theta_2}$
      & ${\mathsf j}\,e^{-{\mathsf j}(\theta_1 + \theta_2)}$
      & $-e^{{\mathsf j}\theta_1}$
      & $-e^{{\mathsf j}\theta_2}$
      & $-{\mathsf j}\,e^{-{\mathsf j}(\theta_1 + \theta_2)}$ \\
      \bottomrule
    \end{tabular}
  \end{table}
\end{example}

\begin{example}\label{exmpl:7b}
  \begin{figure}[!tb]
    \centering
    \begin{tikzpicture}[
        x=0.8cm, y=0.8cm,
        sensor/.style={circle, fill=black, draw=black, inner sep=0pt, minimum size=5pt},
        removed/.style={circle, fill=white, draw=black, inner sep=0pt, minimum size=5pt, line width=0.5pt},
      ]
      \draw[gray!60, thin] (0.5,0) -- (9.5,0);

      \foreach \p in {1,2,4,5,7,8,9}{
        \node[sensor] at (\p,0) {};
      }

      \foreach \p in {3,6}{
        \node[removed] at (\p,0) {};
      }

      \foreach \p in {1,...,9}{
        \node[anchor=north, font=\footnotesize] at (\p,-0.2) {$\p$};
      }
    \end{tikzpicture}
     \caption{Thinned ULA with $\mathcal{I} = \{1, 2, 4, 5, 7, 8, 9\}$, obtained by
    removing positions $3$ and $6$ from a nine-element ULA
    (Example~\ref{exmpl:7b}\,). It is \emph{unambiguous} for $L \leq 5$, although
  it contains no five-element contiguous subarray.}
    \label{fig:holes_at_3_6}
  \end{figure}

Consider the thinned ULA in Fig.~\ref{fig:holes_at_3_6} with
$M_{\mathcal I}=7$ sensors, holes at $i_1=3$ and $i_2=6$, and $L=5$
sources, i.e., $\mathcal I={1,2,4,5,7,8,9}$. The complementary
submatrix is
\begin{equation}
\mathcal{T}_{\bar{\mathcal I}}(\boldsymbol g)
=
\begin{bmatrix}
g_2 & g_0^* \\
g_1 & g_1^* \\
g_0 & g_2^* \\
0   & g_2
\end{bmatrix}.
\end{equation}
Setting $g_0=1$ without loss of generality, a row swap followed by
elimination below the first pivot yields
\begin{equation}\label{eq:TbarI_elim}
\mathcal{T}_{\bar{\mathcal I}}(\boldsymbol g)
\;\overset{\mathrm{row}}{\equiv}\;
    \begin{bmatrix}
      1   & g_2^* \\
      g_2 & 1 \\
      g_1 & g_1^* \\
      0   & g_2
    \end{bmatrix}.
\end{equation}
  The first column has the nonzero pivot $1$, and the second column below
  the first row vanishes entirely if and only if
  \begin{equation}
    1 - |g_2|^2 = 0,
    \qquad
    g_2 = 0,
    \qquad
    g_1^* - g_1\, g_2^* = 0
  \end{equation}
  all hold simultaneously. The first two conditions require both
  $|g_2| = 1$ and $g_2 = 0$, which is impossible. Hence the second column
  of~\eqref{eq:TbarI_elim} cannot vanish below the pivot, and
  \begin{equation}
    \operatorname{rank}\!\big(\mathcal{T}_{\bar{\mathcal{I}}}(\boldsymbol{g})\big) = 2
  \end{equation}
  for all $g_1, g_2 \in \mathbb{C}$. By Theorem~\ref{thm:rank_condition}\,,
  the array manifold ${\boldsymbol A}_{\mathcal{I}}(z_1, \ldots, z_5)$ has
  full column rank for every set of distinct generators, and the array is
  therefore unambiguous for $L = 5$ sources.
\end{example}
\begin{example}\label{exmpl:8a}
  \begin{figure}[!tb]
    \centering
    \begin{tikzpicture}[
        x=0.8cm, y=0.8cm,
        sensor/.style={circle, fill=black, draw=black, inner sep=0pt, minimum size=5pt},
        removed/.style={circle, fill=white, draw=black, inner sep=0pt, minimum size=5pt, line width=0.5pt},
      ]
      \draw[gray!60, thin] (0.5,0) -- (9.5,0);

      \foreach \p in {1,7,8,9}{
        \node[sensor] at (\p,0) {};
      }

      \foreach \p in {2,3,4,5,6}{
        \node[removed] at (\p,0) {};
      }

      \foreach \p in {1,...,9}{
        \node[anchor=north, font=\footnotesize] at (\p,-0.2) {$\p$};
      }
    \end{tikzpicture}
      \caption{Thinned ULA with $\mathcal{I} = \{1, 7, 8, 9\}$ obtained by removing
    $\bar{\mathcal{I}} = \{2, 3, 4, 5, 6\}$ from a nine-element ULA
    (Example~\ref{exmpl:8a}\,). It is \emph{unambiguous} for $L \leq 3$, but
  contains a dense three-element subarray.}
    \label{fig:I_1_7_8_9}
  \end{figure}

  Consider the design of a thinned ULA with $M_{\mathcal{I}} = 4$ sensors,
  obtained by thinning a nine-element ULA, that is unambiguous for $L = 3$
  sources. For $L = 3$, the centro-Hermitian polynomial coefficient vector
  is $\boldsymbol{g} = [g_0, g_1, g_1^*, g_0^*]^{\mathsf{T}}$, and the
  Toeplitz matrix takes the form
  \begin{equation}\label{eq:example8a}
    \mathcal{T}(\boldsymbol{g}) =
    \begin{bmatrix}
      g_0 & g_1 & g_1^* & g_0^* & 0     & 0     & 0     & 0     & 0     \\
      0   & g_0 & g_1   & g_1^* & g_0^* & 0     & 0     & 0     & 0     \\
      0   & 0   & g_0   & g_1   & g_1^* & g_0^* & 0     & 0     & 0     \\
      0   & 0   & 0     & g_0   & g_1   & g_1^* & g_0^* & 0     & 0     \\
      0   & 0   & 0     & 0     & g_0   & g_1   & g_1^* & g_0^* & 0     \\
      0   & 0   & 0     & 0     & 0     & g_0   & g_1   & g_1^* & g_0^*
    \end{bmatrix}.
  \end{equation}
  By Corollary~\ref{crl:corollary2}\,, unambiguity is guaranteed if the array is
  thinned so that a maximum-consecutive-lag array is obtained, here
  $\mathcal{I} = \{1, 7, 8, 9\}$, as depicted in Fig.~\ref{fig:I_1_7_8_9}\,.
  The corresponding complementary submatrix is
  \begin{equation}\label{eq:TbarI_exmpl8a}
    \mathcal{T}_{\bar{\mathcal{I}}}(\boldsymbol{g}) =
    \begin{bmatrix}
      g_1 & g_1^* & g_0^* & 0     & 0     \\
      g_0 & g_1   & g_1^* & g_0^* & 0     \\
      0   & g_0   & g_1   & g_1^* & g_0^* \\
      0   & 0     & g_0   & g_1   & g_1^* \\
      0   & 0     & 0     & g_0   & g_1   \\
      0   & 0     & 0     & 0     & g_0
    \end{bmatrix}.
  \end{equation}
  The $5 \times 5$ submatrix formed by the last five rows
  of~\eqref{eq:TbarI_exmpl8a} is upper triangular with diagonal entries
  $g_0$, so its determinant is $g_0^5 \neq 0$. Hence
  $\mathcal{T}_{\bar{\mathcal{I}}}(\boldsymbol{g})$ has full column rank for
  all $\boldsymbol{g}$, and by Theorem~\ref{thm:rank_condition} the array
  manifold is unambiguous for $L = 3$. The same conclusion extends to
  $L < 3$, since the retained set contains the three-element contiguous
  ULA subarray $\{7, 8, 9\}$, which by the Vandermonde determinant theorem
  is unambiguous for any $L \leq 3$ sources. The result is also a direct
  instance of Corollary~\ref{crl:corollary2}\,.

  However, the array contains a dense three-element subarray at positions
  $\{7, 8, 9\}$, which makes it susceptible to mutual coupling
  effects~\cite{7458883LiuPP}. This motivates the alternative design of
  Example~\ref{exmpl:8b} below.
\end{example}
\begin{example}\label{exmpl:8b}
  \begin{figure}[!tb]
    \centering
    \begin{tikzpicture}[
        x=0.8cm, y=0.8cm,
        sensor/.style={circle, fill=black, draw=black, inner sep=0pt, minimum size=5pt},
        removed/.style={circle, fill=white, draw=black, inner sep=0pt, minimum size=5pt, line width=0.5pt},
      ]
      \draw[gray!60, thin] (0.5,0) -- (9.5,0);
      \foreach \p in {1,2,7,9}{ \node[sensor] at (\p,0) {}; }
      \foreach \p in {3,4,5,6,8}{ \node[removed] at (\p,0) {}; }
      \foreach \p in {1,...,9}{
        \node[anchor=north, font=\footnotesize] at (\p,-0.2) {$\p$};
      }
    \end{tikzpicture}
   \caption{Thinned ULA with $\mathcal{I} = \{1, 2, 7, 9\}$ obtained by removing
    $\bar{\mathcal{I}} = \{3, 4, 5, 6, 8\}$ from a nine-element ULA
    (Example~\ref{exmpl:8b}\,). It contains a single minimum-baseline pair,
  between sensors $1$ and $2$, and is \emph{unambiguous} for $L \leq 3$.}
  \label{fig:single_min_baseline}
  \end{figure}

As an alternative to Example~\ref{exmpl:8a}\,, we seek a design with the same
  $M_{\mathcal{I}} = 4$ sensors and $L = 3$ sources in which no two retained
  sensors other than a single pair are separated by the minimum baseline
  $\lambda/2$, so that the dense subarray of Fig.~\ref{fig:I_1_7_8_9} is
  avoided.

  A guideline for selecting the removal set is provided by the special
  coefficient choice $g_0 = 1$, $g_1 = 0$ in~\eqref{eq:example8a}, under which
  columns $2$, $5$ and $8$ of $\mathcal{T}(\boldsymbol{g})$ are linearly
  dependent, so that these positions cannot be jointly contained in the removal
  set $\bar{\mathcal{I}}$. Combined with the requirement that the first and last
  sensors at positions $1$ and $9$ be retained, this leads to the candidate
  configuration $\mathcal{I} = \{1, 2, 7, 9\}$ shown in
  Fig.~\ref{fig:single_min_baseline}\,, in which only one pair of sensors sits at
  the minimum baseline.

   Setting $g_0 = 1$ without loss of generality (by an appropriate rotation of
  the roots) the complementary submatrix for
  this configuration is
  \begin{equation}\label{eq:TbarI_exmpl8b}
    \mathcal{T}_{\bar{\mathcal{I}}}(\boldsymbol{g})
    =
    \begin{bmatrix}
      g_1^* & 1   & 0   & 0   & 0  \\
      g_1   & g_1^* & 1   & 0   & 0  \\
      1   & g_1   & g_1^* & 1   & 0  \\
      0   & 1   & g_1   & g_1^* & 0  \\
      0   & 0   & 1   & g_1   & 1  \\
      0   & 0   & 0   & 1   & g_1^*
    \end{bmatrix},
  \end{equation}
  and the following proposition establishes that it has full column rank for
  every admissible coefficient vector.

  \begin{proposition}\label{prop:exmpl8b}
    The complementary submatrix
    $\mathcal{T}_{\bar{\mathcal{I}}}(\boldsymbol{g})$
    in~\eqref{eq:TbarI_exmpl8b} has full column rank for every
    $g_1 \in \mathbb{C}$, i.e.,
    \begin{equation}
      \operatorname{rank}\!\big(\mathcal{T}_{\bar{\mathcal{I}}}(\boldsymbol{g})\big) = 5 .
    \end{equation}
  \end{proposition}

   \begin{IEEEproof}
    Consider two $5\times5$ minors of
    $\mathcal{T}_{\bar{\mathcal{I}}}(\boldsymbol{g})$. Deleting the sixth and
    fourth rows, respectively, yields the determinants
    \begin{align}
      \Delta_6
      &= g_1^2 - 3g_1(g_1^*)^2 + (g_1^*)^4 + 2g_1^*,
      \label{eq:minor6_exmpl8b}\\
      \Delta_4
      &= -2g_1^2(g_1^*)^2 + g_1(g_1^*)^4 + 4g_1g_1^* - 2(g_1^*)^3 - 1.
      \label{eq:minor4_exmpl8b}
    \end{align}
    If~\eqref{eq:TbarI_exmpl8b} were rank deficient, all its $5\times5$ minors
    would vanish and, in particular, $\Delta_6 = \Delta_4 = 0$. Define
    \begin{equation}
      \rho \triangleq |g_1| \geq 0,
      \qquad
      \tau \triangleq (g_1^*)^3 .
    \end{equation}
    Multiplying $\Delta_6 = 0$ by $(g_1^*)^2$ gives
    \begin{equation}\label{eq:minor6_rq_exmpl8b}
      \tau^2 + (2 - 3\rho^2)\,\tau + \rho^4 = 0,
    \end{equation}
    whereas $\Delta_4 = 0$ is equivalent to
    \begin{equation}\label{eq:minor4_rq_exmpl8b}
      (\rho^2 - 2)\,\tau - 2\rho^4 + 4\rho^2 - 1 = 0 .
    \end{equation}
    For $\rho^2 = 2$, the left-hand side of~\eqref{eq:minor4_rq_exmpl8b} equals
    $-1$, so $\rho^2 \neq 2$. Hence
    \begin{equation}
      \tau = \frac{2\rho^4 - 4\rho^2 + 1}{\rho^2 - 2} \in \mathbb{R},
    \end{equation}
    i.e., $(g_1^*)^3$ is real, and since $|\tau| = \rho^3$ it follows that
    \begin{equation}
      \tau = \pm\,\rho^3 .
    \end{equation}
    For $\tau = \rho^3$, \eqref{eq:minor6_rq_exmpl8b} factors as
    \begin{equation}
      \rho^3(\rho - 2)(\rho^2 - \rho - 1) = 0 .
    \end{equation}
    The possibilities $\rho = 0$ and $\rho = 2$ are immediately ruled out
    by~\eqref{eq:minor4_rq_exmpl8b}, whose left-hand side then equals $-1$. If
    $\rho^2 - \rho - 1 = 0$, substitution into~\eqref{eq:minor4_rq_exmpl8b}
    reduces its left-hand side to $-\rho$, which is nonzero. Thus
    $\tau = \rho^3$ is impossible.

    Similarly, for $\tau = -\rho^3$, \eqref{eq:minor6_rq_exmpl8b} factors as
    \begin{equation}
      \rho^3(\rho + 2)(\rho^2 + \rho - 1) = 0 .
    \end{equation}
    Since $\rho \geq 0$, the only possibilities are $\rho = 0$, which is again
    ruled out by~\eqref{eq:minor4_rq_exmpl8b}, and $\rho^2 + \rho - 1 = 0$, for
    which its left-hand side reduces to $\rho \neq 0$. Hence
    $\tau = -\rho^3$ is also impossible.

    Therefore $\Delta_4$ and $\Delta_6$ cannot vanish simultaneously for any
    $g_1 \in \mathbb{C}$, so at least one $5\times5$ minor of
    $\mathcal{T}_{\bar{\mathcal{I}}}(\boldsymbol{g})$ is nonzero.
  \end{IEEEproof}

  By Theorem~\ref{thm:rank_condition}\,, the array $\mathcal{I} = \{1, 2, 7, 9\}$
  is therefore unambiguous for $L = 3$ sources, and the same conclusion extends
  to $L < 3$. In contrast to the maximum-consecutive-lag array of
  Example~\ref{exmpl:8a}\,, this design achieves unambiguity with only a single
  minimum-baseline pair.
\end{example}
\begin{example}\label{exmpl:9b}
  \begin{figure}[!tb]
    \centering
    \begin{tikzpicture}[
        x=0.36cm, y=0.36cm,
        sensor/.style={circle, fill=black, draw=black, inner sep=0pt, minimum size=4pt},
        removed/.style={circle, fill=white, draw=black, inner sep=0pt, minimum size=4pt, line width=0.4pt},
      ]
      \draw[gray!60, thin] (0.5,0) -- (23.5,0);

      \foreach \i in {1,3,4,5,6,7,8,9,10,12,13,14,15,20,23}{
        \node[sensor] at (\i,0) {};
      }

      \foreach \i in {2,11,16,17,18,19,21,22}{
        \node[removed] at (\i,0) {};
      }

      \foreach \i in {1,5,10,15,20,23}{
        \node[anchor=north, font=\tiny] at (\i,-0.3) {\i};
      }
    \end{tikzpicture}
    \caption{Thinned ULA with $M_{\mathcal{I}} = 15$ sensors obtained from a
    $23$-element ULA. Filled circles denote the retained sensors at positions
    $\mathcal{I} = \{1, 3, \ldots, 10, 12, \ldots, 15, 20, 23\}$; open circles
    mark the removed positions
    $\bar{\mathcal{I}} = \{2, 11, 16, \ldots, 19, 21, 22\}$. The configuration
    is \emph{unambiguous} for $L = 10$ (Example~\ref{exmpl:9b}\,) and
  \emph{ambiguous} for $L = 11$ (Example~\ref{exmpl:9a}\,).}
  \label{fig:thinned_ula}
  \end{figure}
  Consider the thinned ULA depicted in Fig.~\ref{fig:thinned_ula}\,, with
  $M_{\mathcal{I}} = 15$ sensors obtained from a $23$-element ULA by
  removing the sensors at positions
  $\bar{\mathcal{I}} = \{2, 11, 16, 17, 18, 19, 21, 22\}$, and $L = 10$
  sources. The sensor positions were selected  in agreement with the mutually exclusive index pairs listed in Table \ref{tab:Example11}\,, which are obtained from Corollary \ref{crl:corollary4b}\,.

 Each entry $\mathcal{S}_{p,q} = \{q,\, L - p + q\}$ of
  Table~\ref{tab:Example11} is a pair of grid positions of which at most one may
  be omitted. A layout is therefore admissible in the sense of
  Corollary~\ref{crl:corollary4b} if no pair is contained in
  $\bar{\mathcal{I}}$ as a whole; individual members may be, and here typically
  are, removed. 

For the present layout, for instance, the pair
  $\mathcal{S}_{0,6} = \{6, 16\}$ has $16 \in \bar{\mathcal{I}}$ but
  $6 \in \mathcal{I}$, and likewise $\mathcal{S}_{0,7} = \{7, 17\}$,
  $\mathcal{S}_{0,8} = \{8, 18\}$ and $\mathcal{S}_{0,9} = \{9, 19\}$ each have
  exactly one member removed. A direct check of all twelve pairs confirms that
  none is fully contained in
  $\bar{\mathcal{I}} = \{2, 11, 16, 17, 18, 19, 21, 22\}$, which is what makes the
  layout of Fig.~\ref{fig:thinned_ula} a candidate for an unambiguous design at
  $L = 10$. That for $L = 10$ the candidate is indeed unambiguous is established by
  Proposition~\ref{prop:exmpl9b} below.
  \begin{table}[!tb]
    \centering
    \caption{Mutually exclusive index pairs ${\mathcal S}_{p,q}$ for $M=23$, $L=10$ according to Corollary \ref{crl:corollary4b}\,, as used in Example~\ref{exmpl:9b}\,. No pair is fully contained in $\bar{\mathcal I} = \{2, 11, 16, 17, 18, 19, 21, 22\}$.}
    \label{tab:Example11}
    {\setlength{\tabcolsep}{3pt}
      \begin{tabular}{rrrr}
        \hline
        ${\mathcal S}_{0,4} = \{4, 14\}$  & ${\mathcal S}_{0,5} = \{5, 15\}$ & ${\mathcal S}_{0,6} = \{6, 16\}$ & ${\mathcal S}_{0,7} = \{7, 17\}$ \\
        ${\mathcal S}_{0,8} = \{8, 18\}$  & ${\mathcal S}_{0,9} = \{9, 19\}$ & ${\mathcal S}_{0,10} = \{10, 20\}$ & ${\mathcal S}_{1,6} = \{6, 15\}$ \\
        ${\mathcal S}_{1,7} = \{7, 16\}$  & ${\mathcal S}_{1,8} = \{8, 17\}$ & ${\mathcal S}_{1,9} = \{9, 18\}$ & ${\mathcal S}_{2,8} = \{8, 16\}$ \\
        \hline
    \end{tabular}}
  \end{table}
The coefficient vector is now
  $\boldsymbol{g} = [g_0, g_1, g_2, g_3, g_4, g_5, g_4^*, g_3^*, g_2^*,
  g_1^*, g_0^*]^{\mathsf{T}}$, where $g_1, \ldots, g_4 \in \mathbb{C}$ and
  $g_5 \in \mathbb{R}$, and ${\mathcal{T}}(\boldsymbol{g}) \in
  \mathbb{C}^{13 \times 23}$ is given by~\eqref{def:ToeplitzEven} with
  $K = 5$. Setting $g_0 = 1$ without loss of generality and selecting the columns indexed
  by $\bar{\mathcal{I}}$ yields the complementary submatrix
  $\mathcal{T}_{\bar{\mathcal{I}}}(\boldsymbol{g})$ shown
  in~\eqref{eq:TbarI_exmpl9b}.

\begin{figure}[!tb]
  \normalsize
  \begin{equation}\label{eq:TbarI_exmpl9b}
    \mathcal{T}_{\bar{\mathcal{I}}}(\boldsymbol{g}) =
    \setlength{\arraycolsep}{2pt}
    \begin{bmatrix}
      g_1 & 1     & 0     & 0     & 0     & 0     & 0     & 0\\
      1   & g_1^* & 0     & 0     & 0     & 0     & 0     & 0\\
      0   & g_2^* & 0     & 0     & 0     & 0     & 0     & 0\\
      0   & g_3^* & 0     & 0     & 0     & 0     & 0     & 0\\
      0   & g_4^* & 0     & 0     & 0     & 0     & 0     & 0\\
      0   & g_5   & 1     & 0     & 0     & 0     & 0     & 0\\
      0   & g_4   & g_1^* & 1     & 0     & 0     & 0     & 0\\
      0   & g_3   & g_2^* & g_1^* & 1     & 0     & 0     & 0\\
      0   & g_2   & g_3^* & g_2^* & g_1^* & 1     & 0     & 0\\
      0   & g_1   & g_4^* & g_3^* & g_2^* & g_1^* & 0     & 0\\
      0   & 1     & g_5   & g_4^* & g_3^* & g_2^* & 1     & 0\\
      0   & 0     & g_4   & g_5   & g_4^* & g_3^* & g_1^* & 1\\
      0   & 0     & g_3   & g_4   & g_5   & g_4^* & g_2^* & g_1^*
    \end{bmatrix}
  \end{equation}
  \hrulefill
  \vspace*{4pt}
\end{figure}

  \begin{proposition}\label{prop:exmpl9b}
    The complementary submatrix
    $\mathcal{T}_{\bar{\mathcal{I}}}(\boldsymbol{g})$
    in~\eqref{eq:TbarI_exmpl9b} has full column rank for every choice of
    coefficients, i.e.,
    \begin{equation}
      \operatorname{rank}\!\big(\mathcal{T}_{\bar{\mathcal{I}}}(\boldsymbol{g})\big) = 8.
    \end{equation}
  \end{proposition}
  \begin{IEEEproof}
    We show that the only null vector
    $\boldsymbol{\alpha} = [\alpha_1, \ldots, \alpha_8]^{\mathsf{T}} \in
    \mathbb{C}^8$ of $\mathcal{T}_{\bar{\mathcal{I}}}(\boldsymbol{g})$ is the
    zero vector. The system
    $\mathcal{T}_{\bar{\mathcal{I}}}(\boldsymbol{g})\, \boldsymbol{\alpha}
    = \boldsymbol{0}$ reads
    {
    \begin{subequations}\label{eq:TbarI-alpha-system}
 \allowdisplaybreaks
      \begin{align}
        g_1\alpha_1 + \alpha_2 &= 0, \label{eq:TbarI-alpha-1}\\
        \alpha_1 + g_1^*\alpha_2 &= 0, \label{eq:TbarI-alpha-2}\\
        g_2^*\alpha_2 = g_3^*\alpha_2 = g_4^*\alpha_2 &= 0, \label{eq:TbarI-alpha-345}\\
        g_5\alpha_2 + \alpha_3 &= 0, \label{eq:TbarI-alpha-6}\\
        g_4\alpha_2 + g_1^*\alpha_3 + \alpha_4 &= 0, \label{eq:TbarI-alpha-7}\\
        g_3\alpha_2 + g_2^*\alpha_3 + g_1^*\alpha_4 + \alpha_5 &= 0, \label{eq:TbarI-alpha-8}\\
        g_2\alpha_2 + g_3^*\alpha_3 + g_2^*\alpha_4 + g_1^*\alpha_5 + \alpha_6 &= 0, \label{eq:TbarI-alpha-9}\\
        g_1\alpha_2 + g_4^*\alpha_3 + g_3^*\alpha_4 + g_2^*\alpha_5 + g_1^*\alpha_6 &= 0, \label{eq:TbarI-alpha-10}\\
        \alpha_2 + g_5\alpha_3 + g_4^*\alpha_4 + g_3^*\alpha_5 + g_2^*\alpha_6 + \alpha_7 &= 0, \label{eq:TbarI-alpha-11}\\
        g_4\alpha_3 + g_5\alpha_4 + g_4^*\alpha_5 + g_3^*\alpha_6 + g_1^*\alpha_7 + \alpha_8 &= 0, \label{eq:TbarI-alpha-12}\\
        g_3\alpha_3 + g_4\alpha_4 + g_5\alpha_5 + g_4^*\alpha_6 + g_2^*\alpha_7 + g_1^*\alpha_8 &= 0. \label{eq:TbarI-alpha-13}
      \end{align}
    \end{subequations}
    }
    Equations~\eqref{eq:TbarI-alpha-6}--\eqref{eq:TbarI-alpha-9},
    \eqref{eq:TbarI-alpha-11}, and~\eqref{eq:TbarI-alpha-12} form a
    lower-triangular system in the unknowns
    $\alpha_3, \ldots, \alpha_8$ with unit diagonal: each successive
    equation introduces the next unknown with coefficient~$1$. Hence, once
    $\alpha_2 = 0$ is established, these six equations recursively force
    $\alpha_3 = \cdots = \alpha_8 = 0$, and~\eqref{eq:TbarI-alpha-2} then
    yields $\alpha_1 = 0$. It therefore suffices to show $\alpha_2 = 0$.

    \emph{Case~A: $g_k \neq 0$ for some $k \in \{2, 3, 4\}$.}
    Equation~\eqref{eq:TbarI-alpha-345} gives
    $g_k^*\,\alpha_2 = 0$, and since $g_k^* \neq 0$, we obtain $\alpha_2 = 0$.

    \emph{Case~B: $g_2 = g_3 = g_4 = 0$ and $|g_1| \neq 1$.}
    We consider equations~\eqref{eq:TbarI-alpha-1} and~\eqref{eq:TbarI-alpha-2} that state
    $g_1\,\alpha_1 + \alpha_2 = 0$ and $\alpha_1 + g_1^*\,\alpha_2 = 0$.
    Eliminating $\alpha_1$ yields $(1 - |g_1|^2)\,\alpha_2 = 0$, which forces
    $\alpha_2 = 0$ whenever $|g_1| \neq 1$.

    \emph{Case~C: $g_2 = g_3 = g_4 = 0$ and $|g_1| = 1$.}
    We argue by contradiction. Suppose $\alpha_2 \neq 0$. With
    $g_2 = g_3 = g_4 = 0$, equations~\eqref{eq:TbarI-alpha-6}--%
    \eqref{eq:TbarI-alpha-9} yield recursively
    \begin{align}
      \alpha_3 &= -g_5\,\alpha_2,
      & \alpha_4 &= g_1^*\, g_5\,\alpha_2,\\
      \alpha_5 &= -(g_1^*)^2\, g_5\,\alpha_2,
      & \alpha_6 &= (g_1^*)^3\, g_5\,\alpha_2.
    \end{align}
    Substituting into~\eqref{eq:TbarI-alpha-10} gives
    $\big(g_1 + (g_1^*)^4\, g_5\big)\,\alpha_2 = 0$, hence
    \begin{equation}
      g_1 = -(g_1^*)^4\, g_5.
    \end{equation}
    Taking moduli, $|g_1| = |g_1|^4\, |g_5|$, and using $|g_1| = 1$ gives
    $|g_5| = 1$. Since $g_5 \in \mathbb{R}$, we obtain $g_5 = \pm 1$, so
    $g_5^2 = 1$.

    From~\eqref{eq:TbarI-alpha-11} we then compute
    \begin{equation}
      \alpha_7 = -\alpha_2 - g_5\,\alpha_3
      = -\alpha_2 + g_5^2\,\alpha_2
      = (g_5^2 - 1)\,\alpha_2 = 0,
    \end{equation}
    and from~\eqref{eq:TbarI-alpha-12},
    \begin{equation}
      \alpha_8 = -g_5\,\alpha_4 - g_1^*\,\alpha_7
      = -g_1^*\, g_5^2\,\alpha_2 = -g_1^*\,\alpha_2.
    \end{equation}
    Finally, equation~\eqref{eq:TbarI-alpha-13} reduces to
    $g_5\,\alpha_5 + g_1^*\,\alpha_8 = 0$. Substituting,
    \begin{align}
      - g_5 \,(g_1^*)^2\, g_5\,\alpha_2
      - g_1^* g_1^*\,\alpha_2
      &= -(g_1^*)^2\,(g_5^2 + 1)\,\alpha_2\notag\\
      &= -2\,(g_1^*)^2\,\alpha_2 = 0.
    \end{align}
    Since $|g_1| = 1$, this forces $\alpha_2 = 0$, contradicting the
    assumption $\alpha_2 \neq 0$.

    In every case $\alpha_2 = 0$, and by the opening remark
    $\boldsymbol{\alpha} = \boldsymbol{0}$. Hence
    $\operatorname{rank}\!\big(\mathcal{T}_{\bar{\mathcal{I}}}(\boldsymbol{g})\big) = 8$.
  \end{IEEEproof}

  By Theorem~\ref{thm:rank_condition} and
  Proposition~\ref{prop:exmpl9b}\,, the thinned array
  $\mathcal{I} = \{1, 3, 4, \ldots, 10, 12, \ldots, 15, 20, 23\}$ is
  unambiguous for $L = 10$ sources.
\end{example}
\begin{example}\label{exmpl:9a}
\begin{table}[t]
  \centering
  \caption{Mutually exclusive index pairs
    ${\mathcal S}_{p,q} = \{q,\, L-p+q\}$ for $M = 23$, $L = 11$, according to
    Corollary~\ref{crl:corollary4b}\,, as used in Example~\ref{exmpl:9a}\,. The
    pair ${\mathcal S}_{0,11} = \{11, 22\}$ is fully contained in
  $\bar{\mathcal I}$ and therefore certifies an ambiguity.}
  \label{tab:Example12}
  {\setlength{\tabcolsep}{5pt}
    \begin{tabular}{ccl}
      \hline
      $p$ & $q$ & ${\mathcal S}_{p,q}$\\
      \hline
      $0$ & $2, \ldots, 11$ & $\{2, 13\}, \{3, 14\}, \ldots, \{11, 22\}$\\
      $1$ & $4, \ldots, 10$ & $\{4, 14\}, \{5, 15\}, \ldots, \{10, 20\}$\\
      $2$ & $6, \ldots, 9$  & $\{6, 15\}, \{7, 16\}, \ldots, \{9, 18\}$\\
      $3$ & $8$             & $\{8, 16\}$\\
      \hline
    \end{tabular}}
\end{table}
  Consider again the thinned ULA of Fig.~\ref{fig:thinned_ula}, but now
  with $L = 11$ sources. For this scenario,
  Corollary~\ref{crl:corollary4b} yields the mutually exclusive index
  pairs displayed in Table~\ref{tab:Example12}\,, which differ from those
  of the previous example with $L = 10$ sources. In particular, both
  indices of the pair $\mathcal{S}_{0, 11} = \{11, 22\}$ are contained
  in $\bar{\mathcal{I}}$, which suggests that the corresponding thinned
  steering matrix ${\boldsymbol A}_{\mathcal{I}}(z_1, \ldots, z_{11})$
  is not full column rank for every generator set.

  To compute  ambiguity sets for which a rank drop is observed, we use Corollary~\ref{crl:corollary5} and fix one prescribed
  generator at $\acute{z}_1 = 1$. The cofactor polynomial has degree
  $L - L' = 10$ with centro-Hermitian coefficient vector
  $\boldsymbol{f} = [f_0, f_1, f_2, f_3, f_4, f_5, f_4^*, f_3^*, f_2^*,
  f_1^*, f_0^*]^{\mathsf{T}}$, where $f_\ell \in \mathbb{C}$ for
  $\ell = 0, 1, \ldots, 4$ and $f_5 \in \mathbb{R}$. Then the Toeplitz matrix
  $\mathcal{T}(\boldsymbol{f}) \in \mathbb{C}^{13 \times 23}$ has the same
  form as ${\mathcal{T}}(\boldsymbol{g})$ in Example \ref{exmpl:9b} (with $\boldsymbol{g}$ replaced
  by $\boldsymbol{f}$). Assume without loss
  of generality $|f_0| = 1$.

  Selecting the columns of $\mathcal{T}(\boldsymbol{f})$ indexed by
  $\bar{\mathcal{I}}$ and appending the Vandermonde column for
  $\acute{z}_1 = 1$, the augmented matrix
  in~\eqref{eq:generalizeVandermonde} is displayed in~\eqref{eq:A_exmpl9a}.

\begin{figure}[!t]
  \normalsize
  \begin{equation}\label{eq:A_exmpl9a}
    \mathcal{A}(\boldsymbol{f}, \acute{z}_1) =
    \setlength{\arraycolsep}{2pt}
    \begin{bmatrix}
      f_1 & f_0^{*} & 0 & 0 & 0 & 0 & 0 & 0 & 1\\
      f_0 & f_1^{*} & 0 & 0 & 0 & 0 & 0 & 0 & 1\\
      0 & f_2^{*} & 0 & 0 & 0 & 0 & 0 & 0 & 1\\
      0 & f_3^{*} & 0 & 0 & 0 & 0 & 0 & 0 & 1\\
      0 & f_4^{*} & 0 & 0 & 0 & 0 & 0 & 0 & 1\\
      0 & f_5 & f_0^{*} & 0 & 0 & 0 & 0 & 0 & 1\\
      0 & f_4 & f_1^{*} & f_0^{*} & 0 & 0 & 0 & 0 & 1\\
      0 & f_3 & f_2^{*} & f_1^{*} & f_0^{*} & 0 & 0 & 0 & 1\\
      0 & f_2 & f_3^{*} & f_2^{*} & f_1^{*} & f_0^{*} & 0 & 0 & 1\\
      0 & f_1 & f_4^{*} & f_3^{*} & f_2^{*} & f_1^{*} & 0 & 0 & 1\\
      0 & f_0 & f_5 & f_4^{*} & f_3^{*} & f_2^{*} & f_0^{*} & 0 & 1\\
      0 & 0 & f_4 & f_5 & f_4^{*} & f_3^{*} & f_1^{*} & f_0^{*} & 1\\
      0 & 0 & f_3 & f_4 & f_5 & f_4^{*} & f_2^{*} & f_1^{*} & 1
    \end{bmatrix}
  \end{equation}
  \hrulefill
  \vspace*{4pt}
\end{figure}
  By Corollary~\ref{crl:corollary5}\,, the array manifold is ambiguous if and
  only if $\mathcal{A}(\boldsymbol{f}, \acute{z}_1)$ is rank deficient,
  i.e., if there exists a nonzero null vector
  $\boldsymbol{\alpha} = [\alpha_1, \alpha_2, \ldots, \alpha_9]^{\mathsf{T}}
  \in \mathbb{C}^9$ satisfying
  \begin{equation}\label{eq:Aalphaequalzero}
    \mathcal{A}(\boldsymbol{f}, \acute{z}_1)\, \boldsymbol{\alpha}
    = \boldsymbol{0}.
  \end{equation}
  The thirteen row equations of~\eqref{eq:Aalphaequalzero} read
  \begin{subequations}
    \allowdisplaybreaks
    \begin{alignat}{2}
      f_1\alpha_1 + f_0^*\alpha_2 + \alpha_9 &= 0, &&\label{r1}\\
      f_0\alpha_1 + f_1^*\alpha_2 + \alpha_9 &= 0, &&\label{r2}\\
      f_2^*\alpha_2 + \alpha_9 &= 0, &&\label{r3}\\
      f_3^*\alpha_2 + \alpha_9 &= 0, &&\label{r4}\\
      f_4^*\alpha_2 + \alpha_9 &= 0, &&\label{r5}\\
      f_5\alpha_2 + f_0^*\alpha_3 + \alpha_9 &= 0, &&\label{r6}\\
      f_4\alpha_2 + f_1^*\alpha_3 + f_0^*\alpha_4
      + \alpha_9 &= 0, &&\label{r7}\\
      f_3\alpha_2 + f_2^*\alpha_3 + f_1^*\alpha_4
      + f_0^*\alpha_5 + \alpha_9 &= 0, &&\label{r8}\\
      f_2\alpha_2 + f_3^*\alpha_3 + f_2^*\alpha_4
      + f_1^*\alpha_5 & \notag\\
      \quad {}+ f_0^*\alpha_6 + \alpha_9 &= 0, &&\label{r9}\\
      f_1\alpha_2 + f_4^*\alpha_3 + f_3^*\alpha_4
      + f_2^*\alpha_5 & \notag\\
      \quad {}+ f_1^*\alpha_6 + \alpha_9 &= 0, &&\label{r10}\\
      f_0\alpha_2 + f_5\alpha_3 + f_4^*\alpha_4
      + f_3^*\alpha_5 & \notag\\
      \quad {}+ f_2^*\alpha_6 + f_0^*\alpha_7
      + \alpha_9 &= 0, &&\label{r11}\\
      f_4\alpha_3 + f_5\alpha_4 + f_4^*\alpha_5
      + f_3^*\alpha_6 & \notag\\
      \quad {}+ f_1^*\alpha_7 + f_0^*\alpha_8
      + \alpha_9 &= 0, &&\label{r12}\\
      f_3\alpha_3 + f_4\alpha_4 + f_5\alpha_5
      + f_4^*\alpha_6 & \notag\\
      \quad {}+ f_2^*\alpha_7 + f_1^*\alpha_8
      + \alpha_9 &= 0. &&\label{r13}
    \end{alignat}
  \end{subequations}

  Rows~\eqref{r3}-\eqref{r5} immediately yield
  \begin{equation}\label{eq:f2isf3isf4is_alpha9}
    f_2^*\,\alpha_2 = f_3^*\,\alpha_2 = f_4^*\,\alpha_2 = -\alpha_9.
  \end{equation}
  If $\alpha_2\! = \!0$, then $\alpha_9\! =\! 0$, and back-substitution
  through~\eqref{r6}-\eqref{r12} forces $\alpha_3 = \cdots = \alpha_8 = 0$
  and then $\alpha_1 = 0$, yielding only the trivial solution
  $\boldsymbol{\alpha} = \boldsymbol{0}$. The case $\alpha_2 \neq 0$ is
  therefore the only one capable of producing a rank drop; we normalize
  $\alpha_2 = 1$, so that~\eqref{eq:f2isf3isf4is_alpha9} gives
  \begin{equation}\label{eq:f2_f3_f4_alpha}
    f_2 = f_3 = f_4 = -\alpha_9^*.
  \end{equation}
  The boundary subcase $\alpha_9 = 0$ was treated in the previous Example~\ref{exmpl:9b}\,, where it was shown to admit only the trivial null vector. We
  therefore assume $\alpha_9 \neq 0$ in what follows, and distinguish two
  subcases based on the relation between $f_0$ and $f_1$.
  \begin{table}[t]
    \centering
    \caption{The ambiguity sets of the thinned array of
      Fig.~\ref{fig:thinned_ula} for $L = 11$ sources, obtained from
      Corollary~\ref{crl:corollary5} with the prescribed generator
      $\acute{z}_1 = 1$. Listed are the free parameters of the cofactor
      coefficient vector
      $\boldsymbol{f} = [f_0, f_1, f_2, f_3, f_4, f_5, f_4^*, f_3^*, f_2^*,
      f_1^*, f_0^*]^{\mathsf{T}}$, where $f_0 = e^{{\mathsf j}\phi_0}$,
      $f_2 = f_3 = f_4$ by~\eqref{eq:f2_f3_f4_alpha}, and $f_5 \in \mathbb{R}$;
      for all three sets the representative with $\phi_0 = 0$ is chosen, for
      which all coefficients are real. Sets~B and~C belong to the branch
      $f_0 = f_1$. The generators are $\acute{z}_1 = 1$ together with the ten
      roots of $F(z) = \sum_{\ell=0}^{10} f_\ell z^\ell$, all distinct and on
      the unit circle, and each set yields a rank deficit $D = 1$. Set~C is the
      set of $11$-th roots of unity predicted by
    Corollary~\ref{crl:corollary4b} through the pair $\mathcal{S}_{0,11} = \{11, 22\}$.}
    \label{tab:ambiguity_sets_L11}
    {\setlength{\tabcolsep}{4pt}
      \renewcommand{\arraystretch}{1.15}
      \begin{tabular}{crrrr}
        \toprule
        & $\phi_0$ & $f_1$ & $f_2$ & $f_5$\\
        \midrule
        A & $0$ & $-1.039973406495$ & $-0.039973406495$ & $\phantom{-}0.662198920517$\\
        B & $0$ & $\phantom{-}1$    & $-0.395336994467$ & $-1.112009743749$\\
        C & $0$ & $\phantom{-}1$    & $\phantom{-}1$    & $\phantom{-}1$\\
        \bottomrule
    \end{tabular}}
  \end{table}

  \emph{Case 1: $f_0 \neq f_1$.}
  With $\alpha_2 = 1$, rows~\eqref{r1} and~\eqref{r2} become
  \begin{align}
    f_1\,\alpha_1 + f_0^* + \alpha_9 &= 0,\\
    f_0\,\alpha_1 + f_1^* + \alpha_9 &= 0.
  \end{align}
  Eliminating $\alpha_1$ between these two equations and using
  $f_0 f_0^* = 1$ yields
  \begin{equation}\label{eq:alpha9of_f1_f0}
    \alpha_9 = -\frac{|f_1|^2 - 1}{f_1 - f_0}.
  \end{equation}
  Thus $\alpha_9$ is determined by $(f_0, f_1)$ and is not an independent
  parameter. Once $f_0$, $f_1$, and $f_5$ are fixed,
  rows~\eqref{r6}--\eqref{r9}, \eqref{r11}, and~\eqref{r12} determine
  $\alpha_3, \alpha_4, \ldots, \alpha_8$ recursively. The two remaining
  rows~\eqref{r10} and~\eqref{r13} then produce two complex residual
  conditions
  \begin{equation}\label{eq:F10_F13}
    F_{10}(f_0, f_1, f_5) = 0,
    \qquad
    F_{13}(f_0, f_1, f_5) = 0,
  \end{equation}
  where the functional dependence on $f_1^*$ is implicit through conjugate
  relations. With the parameterization $f_0 = e^{{\mathsf j}\phi_0}$ and
  $f_1 = \operatorname{Re}(f_1) + {\mathsf j}\,\operatorname{Im}(f_1)$,
  condition~\eqref{eq:F10_F13} becomes four real equations in the four
  real unknowns $(\phi_0, \operatorname{Re}(f_1), \operatorname{Im}(f_1),
  f_5)$, which are solved numerically. For each solution the coefficients
  $f_2 = f_3 = f_4 = -\alpha_9^*$ follow from~\eqref{eq:alpha9of_f1_f0}
  via~\eqref{eq:f2_f3_f4_alpha}.

  Rank deficiency of $\mathcal{A}(\boldsymbol{f}, \acute{z}_1)$ is necessary
  but not sufficient for an ambiguity: as noted
  after~\eqref{eq:selfinversive}, centro-Hermitian coefficients do not
  guarantee that the roots of $P_{\boldsymbol{g}}(z)$ lie on the unit circle,
  and by (A2) the generators must in addition be distinct. Every solution is
  therefore screened by computing the $L$ roots and retaining only those with
  $L$ distinct unimodular generators, as in Examples~\ref{exmpl:3}
  and~\ref{exmpl:4a}\,.\,\footnote{The screening cannot be omitted. In the branch $f_0 = f_1$ of
  Case~2 below, the solution $\phi_0 = 0$,
  $\operatorname{Re}(f_2) = 1.474626617563$, $\operatorname{Im}(f_2) = 0$,
  $f_5 = 3.581545957939$ renders $\mathcal{A}(\boldsymbol{f}, \acute{z}_1)$
  rank deficient, but only five of the eleven roots of $P_{\boldsymbol{g}}(z)$
  are unimodular. The remaining six form three reciprocal pairs, with moduli
  $0.834$ and $1.199$ for the two real roots and $0.772$ and $1.295$ for the
two conjugate pairs; the solution is therefore discarded.} A further reduction is needed because
  the normalization $\acute{z}_1 = 1$ resolves the rotation invariance only up
  to the choice of which generator is fixed to one, cf.\ the discussion
  following Corollary~\ref{crl:rotation_invariance}\,: an ambiguity set of $L$
  generators reappears as one solution for every such choice, modulo its own
  symmetries. Grouping the retained solutions accordingly, this branch yields
  the single ambiguity set~A of Table~\ref{tab:ambiguity_sets_L11}\,,
  represented by six distinct parameter triples $(\phi_0, f_1, f_5)$, which
  are related pairwise by the conjugation symmetry of the set,
  cf.\ Corollary~\ref{crl:conjugate_invariance}\,. In particular, if
  $(\phi_0, f_1, f_5)$ is a rank-deficient parameter set, then
  $(\pi - \phi_0, -f_1^*, -f_5)$ is one as well, reflecting the combined
  rotation and conjugation symmetries of the underlying Vandermonde structure.

  \emph{Case 2: $f_0 = f_1  \triangleq u$ with $|u| = 1$.}
  In this case, rows~\eqref{r1} and~\eqref{r2} coincide, both reducing to
  $u\,\alpha_1 + u^* + \alpha_9 = 0$, which gives
  $\alpha_1 = -(u^* + \alpha_9)/u$. Unlike Case~1, the coefficient
  $\alpha_9$ is now a free parameter: the relation
  $f_2 = f_3 = f_4 = -\alpha_9^*$ in~\eqref{eq:f2_f3_f4_alpha} relates
  $\alpha_9$ to $f_2$ but neither is determined by $f_0$ and $f_5$ alone. The
  free parameters in this branch are therefore $(\phi_0, f_2, f_5)$, with the
  additional unknowns $\alpha_3, \ldots, \alpha_8$ obtained recursively
  from~\eqref{r6}--\eqref{r9}, \eqref{r11}, and~\eqref{r12}.

  The remaining rows~\eqref{r10} and~\eqref{r13} produce two complex
  residual conditions
  \begin{equation}\label{eq:F10_F13_of_f0_f2_f5}
    F_{10}(f_0, f_2, f_5) = 0,
    \qquad
    F_{13}(f_0, f_2, f_5) = 0,
  \end{equation}
  which, under the parameterization $f_0 = e^{{\mathsf j}\phi_0}$ and
  $f_2 = \operatorname{Re}(f_2) + {\mathsf j}\,\operatorname{Im}(f_2)$, become
  four real equations in the four real unknowns
  $(\phi_0, \operatorname{Re}(f_2), \operatorname{Im}(f_2), f_5)$. Applying
  the same screening and grouping as in Case~1, this branch yields two
  ambiguity sets, listed as~B and~C in
  Table~\ref{tab:ambiguity_sets_L11}\,. Set~B is represented by six parameter
  triples, whereas set~C is represented by the single triple
  $\boldsymbol{f} = \boldsymbol{1}$, since the $11$-th roots of unity are
  invariant under rotation by $e^{{\mathsf j}2\pi/11}$ and the choice of the
  generator fixed to $\acute{z}_1 = 1$ is immaterial. For set~C the associated
  polynomial is $P_{\boldsymbol{f}}(z) = \sum_{\ell = 0}^{10} z^\ell$, which
  together with the prescribed generator $\acute{z}_1 = 1$ gives
  \begin{equation}
    P_{\boldsymbol{g}}(z) = (z - 1)\, P_{\boldsymbol{f}}(z) = z^{11} - 1,
  \end{equation}
  with coefficient vector
  $\boldsymbol{g} = [-1, 0, \ldots, 0, 1]^{\mathsf{T}}$. This is exactly the
  ambiguity certified by the linearly dependent column pair
  $\mathcal{S}_{0, 11} = \{11, 22\}$ of Corollary~\ref{crl:corollary4b}\,, for
  the rotated coefficient vector
  $\grave{\boldsymbol{g}} = [1, 0, \ldots, 0, 1]^{\mathsf{T}}$.

  For each of the three sets it was verified numerically that the $L = 11$
  generators are distinct and lie on the unit circle, and that both
  $\mathcal{A}(\boldsymbol{f}, \acute{z}_1)$ and the generalized Vandermonde
  matrix ${\boldsymbol A}_{\mathcal{I}}(z_1, \ldots, z_{11})$ are rank
  deficient with $D = 1$. Together with Example~\ref{exmpl:9b}\,, we conclude
  that the array of Fig.~\ref{fig:thinned_ula} is unambiguous for $L = 10$ but
  admits the three ambiguity sets of Table~\ref{tab:ambiguity_sets_L11} for
  $L = 11$.
\end{example}
\section{Conclusion and Future Work}\label{sec:ConclusionAndFutureWork}

This paper developed an algebraic framework for complete multi-source
identifiability analysis of thinned ULAs. The central result,
Theorem~\ref{thm:rank_condition}\,, relates rank deficiency of the generalized
Vandermonde steering matrix to that of a thinned Toeplitz matrix derived from
a centro-Hermitian polynomial parameterization. This characterization enables,
for the first time, systematic enumeration of all ambiguity sets for arbitrary
source numbers, apertures, and sensor counts, without searching over candidate
generator sets.

A second contribution, Corollary~\ref{crl:corollary5} (Vandermonde reduction),
relates the thinned-array rank condition to that of an augmented full-ULA
steering matrix with prescribed generators. Besides revealing a structural
link between sparse and fully populated arrays, it provides a constructive
route to ambiguity enumeration, demonstrated for configurations ranging from
small hole patterns in $9$-element ULAs to a $15$-sensor thinning of a
$23$-element ULA. Together with the mutually exclusive sensor-index conditions
in Corollaries~\ref{crl:corollary3}–\ref{crl:corollary4b}\,, the framework
provides both insight into the origin of ambiguities and practical tools for
designing sparse layouts that avoid them.

The results suggest several directions for future research.\
\emph{Sparse-array design guidelines.}
The design corollaries in Section~\ref{sec:ArrayDesign}\,, namely,
Corollary~\ref{crl:corollary3} on forbidden center elements and
Corollaries~\ref{crl:corollary4a} and~\ref{crl:corollary4b} on mutually
exclusive index pairs, follow from canonical choices of the polynomial
coefficient vector $\boldsymbol{g}$, in particular
$\boldsymbol{g}=[1,0,\ldots,0,1]^{\mathsf T}$ and
$\boldsymbol{g}=[1,1,\ldots,1]^{\mathsf T}$. Each identifies a structural
condition under which
$\mathcal{T}_{\bar{\mathcal I}}(\boldsymbol{g})$ loses rank, yielding a
necessary condition on the removal set $\bar{\mathcal I}$ for an unambiguous
array. Systematically exploring other parameter and subset choices may yield
further criteria that exclude additional combinations of removed sensors and
thereby substantially reduce the combinatorial search space for optimal
sparse-array design.

\emph{Optimal sparse-array design under aperture and sensor constraints.}
For fixed aperture $M$ and sensor count $M_{\mathcal I}$, the framework permits
a discrete optimization over admissible thinning patterns to maximize the
steering-matrix spark, equivalently, the highest source number for which the
manifold remains unambiguous, subject to constraints on the longest consecutive
lag and mutual coupling. 
  The trade-off between identifiability and coupling robustness becomes
  binding when the sensor count is large relative to the available aperture,
  so that many pairs of retained sensors necessarily sit at the minimum
  baseline. This regime is typical of holographic and extremely large-scale
  MIMO arrays at mmWave and sub-THz frequencies~\cite{10045715Yuan}, medical
  ultrasound and CMUT/pMUT transducer arrays~\cite{25247666Boujenoui}, and
  compact direction-finding arrays on size-constrained
  platforms~\cite{7458883LiuPP}.

\emph{Coarray-based design of symmetric sparse arrays.}
A complementary direction is the design of symmetric thinned arrays maximizing
the unambiguous source count attainable from the difference
coarray~\cite{10835190gini}. Conventional criteria maximize the number of
consecutive lags, since a contiguous lag set of length $P$ guarantees
identifiability of up to $P-1$ sources from second-order statistics. Coarrays,
however, typically contain many non-consecutive lags beyond this segment whose
contribution to identifiability the consecutive-lag count ignores. The theoretical framework and the tools developed here could bound the largest number of sources for which the
coarray steering matrix retains full column rank, accounting for all lags and
potentially yielding guarantees beyond the conventional one.

\emph{Extension to more general array structures.}
The polynomial parameterization used here is specific to ULAs with
$\lambda/2$-grid sensor positions. Extensions to shift-invariant sparse arrays,
including ESPRIT-compatible geometries~\cite{478488Haardt}, and partly
calibrated sparse arrays~\cite{1025573PesaventoRARE,10636507liu}, where only
a subset of inter-sensor displacements is known, would broaden the
rank-drop characterization to more general practical geometries.
Shift invariance provides algebraic structure analogous to the Toeplitz
parameterization used here. Partly calibrated arrays lead to a similar
separation: the calibrated subarrays constrain part of the manifold while the
remaining displacements are unknown, which parallels the situation in
Corollary~\ref{crl:corollary5}, where $L'$ generators are prescribed and the
remaining $L - L'$ are free.

\emph{Uniqueness of data augmentation.}
Theorem~\ref{thm:rank_condition} also admits a data-augmentation
interpretation. Its proof constructs, for every vector in
$\mathcal{N}(\boldsymbol A_{\mathcal I})$, auxiliary virtual measurements
$\boldsymbol y_{\bar{\mathcal I}}$ at the removed sensors that complete the
sparse observation to a full-ULA observation in
$\mathcal{R}(\boldsymbol A_M)$; cf.\ Fig.~\ref{fig:SchemeRankDeficiency}\,.
Data-augmentation methods in compressed sensing and sparse recovery, including
the interpolation of missing entries or lags in structured Toeplitz or Hankel
matrices and virtual-array
constructions~\cite{books/daglib/0036092Rauhut,7744614Ye,6576276Recht}, rely on
such completions and presuppose their uniqueness. The present framework makes this testable: the completion is unique if and only
if $\mathcal{T}_{\bar{\mathcal I}}(\boldsymbol g)$ has full column rank, and
otherwise the set of admissible completions has dimension $D$.

Beyond array signal processing, the underlying mathematical problem of
determining the spark of generalized Vandermonde matrices, or equivalently their
Kruskal rank, remains open in its full generality. Its connections to super-resolution, finite-rate-of-innovation
sampling, sparse polynomial interpolation, moment problems, algebraic coding
theory, and control theory outlined in the introduction suggest that the
structured rank characterization developed here should prove useful in these
areas as well.
\section*{Acknowledgment}
The author used large language model tools, specifically ChatGPT
[GPT-5.5 Thinking, OpenAI] and Claude
[Claude Opus 5, Anthropic], to assist with symbolic algebraic manipulations and with generating/interfacing Python/SymPy code used to check selected derivations in Section~\ref{sec:Examples}\,.
The AI-generated outputs were independently verified by the author, and the author takes full responsibility for the correctness, originality, and integrity of the manuscript. The author is grateful to Paul Göttlich for his technical and creative support in the preparation of the illustration in Figure~1.
\bibliographystyle{IEEEtran}
\bibliography{references_doi_note}
\end{document}